\documentclass[reqno]{amsart}
\usepackage[utf8]{inputenc}
\usepackage{pmboxdraw}
\usepackage{amsmath}
\usepackage{amssymb}
\usepackage{xcolor}
\usepackage{graphicx}
\usepackage{float}
\usepackage[hidelinks]{hyperref}

\calclayout

\usepackage{cite}
\newtheorem{theorem}{Theorem}
\newtheorem{remark}[theorem]{Remark}
\newtheorem{lemma}[theorem]{Lemma}
\newtheorem{definition}{Definition}
\newtheorem{proposition}[theorem]{Proposition}
\newtheorem{conjecture}[theorem]{Conjecture}
\newtheorem{corollary}[theorem]{Corollary}

\newcommand{\C}{\mathbb{C}}
\newcommand{\N}{\mathbb{N}}

\newcommand{\R}{\mathbb{R}}

\newcommand{\ler}[1]{\left( #1 \right)}
\newcommand{\lesq}[1]{\left[ #1 \right]}
\newcommand{\lers}[1]{\left\{ #1 \right\}}
\newcommand{\hohc}{\cH \otimes \cH^*}

\newcommand{\abs}[1]{\left| #1 \right|}

\newcommand{\norm}[1]{\left|\left|#1\right|\right|}

\newcommand{\be}{\begin{equation}}
\newcommand{\ee}{\end{equation}}
\newcommand{\ba}{\begin{array}}
\newcommand{\ea}{\end{array}}

\newcommand{\fel}{\frac{1}{2}}
\newcommand{\cH}{\mathcal{H}}
\newcommand{\cB}{\mathcal{B}}

\newcommand{\cP}{\mathcal{P}}

\newcommand{\cS}{\mathcal{S}}
\newcommand{\cC}{\mathcal{C}}
\newcommand{\cO}{\mathcal{O}}

\newcommand{\cL}{\mathcal{L}}
\newcommand{\bb}{\mathbf{b}}

\newcommand{\tr}{\mathrm{tr}}
\newcommand{\dd}{\mathrm{d}}

\newcommand{\cA}{\mathcal{A}}

\newcommand{\bC}{\mathbf{C}}
\newcommand{\bP}{\mathbb{P}}

\newcommand{\E}{\mathbb{E}}
\newcommand{\law}{\mathrm{law}}

\newcommand\lh{\cL(\cH)}

\newcommand{\sh}{\cS\ler{\cH}}
\newcommand{\fco}{\mathbf{FCO}}
\newcommand{\fcoh}{\mathbf{FCO}(\cH)}
\newcommand{\ind}[1]{\mathbf{1}_{\left\{#1\right\}}}

\title[$p$-Wasserstein distances and divergences by quantum channels]{Wasserstein distances and divergences of order $p$ by quantum channels}

\author[Gergely Bunth]{Gergely Bunth}
\address{Gergely Bunth, HUN-REN Alfr\'ed R\'enyi Institute of Mathematics\\ Re\'altanoda u. 13-15.\\Budapest H-1053\\ Hungary\\ and Department of Analysis and Operations Research, Institute of Mathematics, Budapest University of Technology and Economics\\ M\H{u}egyetem rkp. 3. \\ Budapest H-1111 \\ Hungary}
\email{bunth.gergely@renyi.hu}

\author[J\'ozsef Pitrik]{J\'ozsef Pitrik}
\address{J\'ozsef Pitrik, HUN-REN Wigner Research Centre for Physics\\ Budapest H-1525, Hungary\\ and HUN-REN Alfr\'ed R\'enyi Institute of Mathematics\\ Re\'altanoda u. 13-15.\\ Budapest H-1053\\ Hungary\\ and Department of Analysis and Operations Research, Institute of Mathematics \\Budapest University of Technology and Economics\\ M\H{u}egyetem rkp. 3. \\ Budapest H-1111\\ Hungary}
\email{pitrik.jozsef@renyi.hu}

\author[Tam\'as Titkos]{Tam\'as Titkos}
\address{Tam\'as Titkos, Corvinus University of Budapest\\ Department of Mathematics\\ Fővám tér 13-15.\\ Budapest H-1093\\Hungary\\ 
and \\ HUN-REN Alfr\'ed R\'enyi Institute of Mathematics\\ Re\'altanoda u. 13-15.\\ Budapest H-1053\\Hungary}
\email{titkos.tamas@renyi.hu}

\author[D\'aniel Virosztek]{D\'aniel Virosztek}
\address{D\'aniel Virosztek, HUN-REN Alfr\'ed R\'enyi Institute of Mathematics\\ Re\'altanoda u. 13-15.\\Budapest H-1053\\ Hungary}
\email{virosztek.daniel@renyi.hu}

\date{}



\thanks{Bunth was supported by the Momentum Program of the Hungarian Academy of Sciences (grant no. LP2021-15/2021); Pitrik was supported by the “Frontline” Research Excellence Programme of the Hungarian National Research, Development and Innovation Office - NKFIH (grant no. KKP133827) and by the Momentum Program of the Hungarian Academy of Sciences (grant no. LP2021-15/2021); Titkos was supported by the Hungarian National Research, Development and Innovation Office - NKFIH (grant no. K115383) and by the Momentum Program of the Hungarian Academy of Sciences (grant no. LP2021-15/2021).; Virosztek was supported by the Momentum program of the Hungarian Academy of Sciences under grant agreement no. LP2021-15/2021, by the Hungarian National Research, Development and Innovation Office (NKFIH) under grant agreement no. Excellence\_151232, and partially supported by the ERC Synergy Grant No. 810115.}

\begin{document}

\begin{abstract}
We introduce a non-quadratic generalization of the quantum mechanical optimal transport problem introduced in [De Palma and Trevisan, Ann. Henri Poincar\'e, {\bf 22} (2021), 3199-3234] where quantum channels realize the transport. Relying on this general machinery, we introduce $p$-Wasserstein distances and divergences and study their fundamental geometric properties. Finally, we prove triangle inequality for quadratic Wasserstein divergences under the sole assumption that an arbitrary one of the states involved is pure, which is a generalization of our previous result in this direction.
\end{abstract}

\maketitle

\section{Introduction}

\subsection{Motivation and main result}

Techniques using the theory of optimal transportation and the nice properties of Wasserstein metric have led to significant advancements in several areas of mathematics, such as probability theory \cite{bgl,Butkovsky}, the study of (stochastic) partial differential equations \cite{Hairer2,Hairer3}, variational analysis \cite{FigalliMaggi1,FigalliMaggi2}, and the geometry of metric spaces \cite{LottVillani,Sturm4,Sturm5,Sturm6}.

Beyond the theoretical advancements in pure mathematics, the nice geometric features of transport-related metrics have given new momentum for research in several different disciplines, ranging from economics \cite{Galichon} and finance to biology \cite{SCHIEBINGER}, not to mention the countless application in applied sciences like biomedical image analysis \cite{bia1,bia2,images1}, data classification \cite{dia1,imageprocessing1}, and machine learning \cite{MC1, SK1,PC,m2,MachineLearning3}.

It is a general phenomenon that concepts and notions that are well-established in the classical commutative world do not have a unique \lq\lq best" extension in the non-commutative world but there are many possible ways of generalization with pros and cons. This is the case concerning optimal transportation as well.
Non-commutative optimal transport is a flourishing research field these days with several essentially different promising approaches such as that of Biane and Voiculescu (free probability) \cite{BianeVoiculescu}, Carlen, Maas, Datta, Rouzé, and Wirth  (dynamical theory) \cite{CarlenMaas-2,CarlenMaas-3,CarlenMaas-4, DattaRouze1, DattaRouze2, Wirth-dual}, Caglioti, Golse, Mouhot, and Paul (quantum many-body problems) \cite{CagliotiGolsePaul, CagliotiGolsePaul-towardsqot, GolseMouhotPaul, GolsePaul-Schrodinger,GolsePaul-OTapproach, GolseTPaul-pseudometrics,GolsePaul-wavepackets}, De Palma and Trevisan (quantum channels) \cite{DPT-AHP,DPT-lecture-notes}, \.Zyczkowski and his collaborators \cite{FriedmanEcksteinColeZyczkowski-MK,ZyczkowskiSlomczynski1,ZyczkowskiSlomczynski2,BistronEcksteinZyczkowski}, and Duvenhage \cite{Duvenhage1, Duvenhage-ext-quantum-det-bal,Duvenhage-quad-Wass-vNA}. Separable quantum Wasserstein distances have also been introduced recently \cite{TothPitrik}, and a novel definition of quantum transport plans has been proposed by Beatty and Stilck Fran\c{c}a in \cite{beatty-franca} where they introduced the corresponding $p$-Wasserstein distances as well. For a comprehensive overview of a substantial part of the above mentioned current approaches to non-commutative optimal transport, the reader is advised to consult the book \cite{OTQS-book}.
\par 
It is an interesting phenomenon that, according to many of the approaches including the one we follow in the present paper \cite{DPT-AHP,DPT-lecture-notes}, the quantum Wasserstein distance of states \emph{is not a genuine metric,} e.g., states may have positive distance from themselves. As a response to this phenomenon, De Palma and Trevisan introduced quadratic quantum Wasserstein divergences \cite{DPT-lecture-notes}, which are appropriately modified versions of quadratic quantum Wasserstein distances, to eliminate self-distances --- see \eqref{eq:quad-qw-div-def} for a precise definition. They conjectured that the divergences defined this way are genuine metrics on quantum state spaces \cite{DPT-lecture-notes}, and this conjecture has been recently justified under certain additional assumptions \cite{BPTV-metric-24}.
\par
In this paper, we propose a non-quadratic generalization of the quantum optimal transport problem relying on quantum channels which was introduced in \cite{DPT-AHP}. In particular, we introduce and study $p$-Wasserstein distances and divergences that arise as (linear combinations of) optimal transport costs with specific cost operators mimicking the classical cost $c(x,y)=\norm{x-y}_p^p.$ Finally, we prove triangle inequality for quadratic Wasserstein divergences under the sole assumption that an arbitrary one of the states involved is pure. This is a generalization of our previous result in this direction \cite{BPTV-metric-24}. It is important to note that ten month after the publication of the first version of this paper, Melchior Wirth proved the triangle inequality for all quadratic quantum Wasserstein divergences in full generality, that is, without any restriction concerning the states involved \cite{wirth-triangle}.

\subsection{Basic notions, notation}
An influential work of De Palma and Trevisan introduced a quantum mechanical counterpart of the classical optimal transport problem with quadratic cost, and also quadratic Wasserstein distances induced by optimal solutions of these transport problems \cite{DPT-AHP}. A key idea of this quantum optimal transport concept is that the transport between quantum states is realized by quantum channels \cite{DPT-AHP,DPT-lecture-notes}.
\par
In the following section, we propose generalizations of the optimal transport problem considered in \cite{DPT-AHP} using cost operators defined by non-quadratic classical cost functions. Let us recall first the classical transportation problem on Euclidean spaces. If $\mu$ and $\nu$ are Borel probability measures on $\R^K$ and $c: \R^K \times \R^K \rightarrow \R$ is a non-negative lower semicontinuous cost function, then the optimal transport problem is to 
\be \label{eq:cl-ot-prob}
\text{minimize } \pi \mapsto \iint_{\R^K \times \R^K} c(x,y) \dd \pi(x,y)
\ee
where $\pi$ runs over all possible couplings of $\mu$ and $\nu.$ We say that $\pi \in \mathrm{Prob}(\R^K \times \R^K)$ is a coupling of $\mu$ and $\nu$ (in notation: $\pi \in \cC(\mu,\nu)$) if the marginals of $\pi$ are $\mu$ and $\nu,$ that is,  
$\iint_{\R^K \times \R^K} f(x) \dd \pi(x,y)=\int_{\R^K} f(x) \dd \mu(x)$ and $\iint_{\R^K \times \R^K} g(y) \dd \pi(x,y)=\int_{\R^K} g(y) \dd \nu(y)$ for all $f,g \in C_b(\R^K).$ A consequence of the tightness (that is, sequential compactness in the weak topology) of $\cC(\mu, \nu)$ and the lower-semicontinuity of $c$ is that there is a coupling (in other words: transport plan) $\pi_0 \in \cC(\mu, \nu)$ that minimizes \eqref{eq:cl-ot-prob}, see, e.g., \cite[Thm. 4.1.]{Villani2}.
\par
Let us recall now those elements of the mathematical formalism of quantum mechanics that we will use throughout this paper. Let $\cH$ be a separable complex Hilbert space. In the sequel, we denote by $\lh^{sa}$ the set of self-adjoint but not necessarily bounded operators on $\cH$, and $\cS(\cH)$ stands for the set of states, that is, the set of positive trace-class operators on $\cH$ with unit trace. The space of all bounded operators on $\cH$ is denoted by $\cB(\cH),$ and we recall that the collection of trace-class operators on $\cH$ is denoted by $\mathcal{T}_1(\cH)$ and defined by $\mathcal{T}_1(\cH)= \lers{X \in \cB(\cH) \, \middle| \, \tr_{\cH}[\sqrt{X^*X}] < \infty}.$ Similarly, $\mathcal{T}_2(\cH)$ stands for the set of Hilbert-Schmidt operators defined by $\mathcal{T}_2(\cH)= \lers{X \in \cB(\cH) \, \middle| \, \tr_{\cH}[X^*X] < \infty}.$
\par
According to the concept introduced in \cite{DPT-AHP}, transport is realized by quantum channels, and the set of all couplings of the quantum states $\rho, \omega \in \cS\ler{\cH}$  --- which we denote by $\cC\ler{\rho, \omega}$ --- is given by
\be \label{eq:q-coup-def}
\cC\ler{\rho, \omega}=\lers{\Pi \in \cS\ler{\cH \otimes \cH^*} \, \middle| \, \tr_{\cH^*} [\Pi]=\omega, \,  \tr_{\cH} [\Pi]=\rho^T},
\ee
where the \emph{transpose} $A^T$ of a linear operator $A$ acting on $\cH$ is a linear operator on the dual space $\cH^*$ defined by the identity $(A^T \eta) (\varphi) \equiv \eta (A \varphi)$ where $\eta \in \cH^*$ and $\varphi \in \mathrm{dom}(A).$
That is, a coupling of $\rho$ and $\omega$ is a state $\Pi$ on $\hohc$ such that 
$$ 
\tr_{\hohc}[\ler{A\otimes I_{\cH^*}} \Pi]=\tr_{\cH} [\omega A]
$$
and
\be \label{eq:part-trace-def}
\tr_{\hohc}\lesq{\ler{I_{\cH} \otimes B^{T}} \Pi}=\tr_{\cH^*} [\rho^T B^T]=\tr_{\cH} [\rho B]
\ee
for all bounded $A, B \in \cL(\cH)^{sa}.$ For further details on the relation between quantum channels sending $\rho$ to $\omega$ and couplings described in \eqref{eq:q-coup-def}, the reader is advised to consult \cite{DPT-AHP} and \cite{DPT-lecture-notes}. Note the clear analogy of the above definition of quantum couplings with the definition of classical couplings that can be phrased as follows: $\pi \in \cP(X^2)$ is a coupling of $\mu \in \cP(X)$ and $\nu \in \cP(X)$ if 
$\iint_{X^2} f(x) \dd \pi (x,y)=\int_{X} f(x) \dd \mu(x)$ and $\iint_{X^2} g(y) \dd \pi (x,y)=\int_{X} g(y) \dd \nu(y)$
for every continuous and bounded function $f,g$ defined on $X.$
We remark that $\cC\ler{\rho,\omega}$ is never empty, because the trivial coupling $\omega\otimes\rho^T$ belongs to $\cC\ler{\rho,\omega}$.
\par
Note that the definition of couplings \eqref{eq:q-coup-def} proposed by De Palma and Trevisan \cite{DPT-AHP} is different from the definition proposed by Golse, Mouhot, Paul \cite{GolseMouhotPaul} in the sense that it involves the dual Hilbert space $\cH^*$ and hence the transpose operation. For a clarification of this difference, see Remark 1 in \cite{DPT-AHP}. For more detail on the latter concept of quantum couplings, the interested reader should consult \cite{CagliotiGolsePaul,CagliotiGolsePaul-towardsqot,GolseMouhotPaul,GolsePaul-Schrodinger,GolsePaul-Nbody, GolsePaul-OTapproach, GolsePaul-meanfieldlimit}. 
\par
\section{Quantum mechanical analogues of the general non-quadratic optimal transport problem on Euclidean spaces}

\par
We propose quantum mechanical analogues of \eqref{eq:cl-ot-prob} which are non-quadratic generalizations of the concept introduced in \cite{DPT-AHP}. The starting point of this proposition was an enlightening discussion with De Palma and Trevisan in September 2024 \cite{DPT-pers-cortona}. The probabilistic interpretation of the classical transport problem will prove useful when introducing the quantum counterpart, hence we recall it now: the task is to minimize $\E[c(X,Y)]$ over all pairs of $\R^K$-valued random vectors $X$ and $Y$ such that $\law(X)=\mu$ and $\law(Y)=\nu.$

\subsection{A quantum mechanical optimal transport problem generalizing \cite{DPT-AHP}} \label{ss:iid-approach}

Let $\cH$ be a separable Hilbert space describing the quantum system we are interested in, and let $\cA=\lers{A_1, \dots, A_K}$ be a finite collection of observables (that is, self-adjoint but possibly unbounded operators) on $\cH.$ Let $E_k$ be the spectral measure of $A_k,$ that is, $A_k=\int_{\R} x \dd E_k(x).$ Let $\Pi$ be a state of the composite system, that is, $\Pi \in \cS(\cH \otimes \cH^*)$ and suppose that $\tr_{\cH^*} [\Pi]=\omega$ and $\tr_{\cH}[\Pi]=\rho^T.$ We consider $K$ copies of the state $\Pi \in \cS(\hohc),$ and for every $k \in \lers{1, \dots, K}$ we measure $A_k$ on the $\cH$ subsystem and $A_k^T$ on the $\cH^*$ subsystem of the $k$th copy. The measurements on different copies are independent from each other.

Let $Y_k$ and $X_k$ denote the real random variables we obtain by measuring $A_k$ on the $\cH$ part and $A_k^T$ on the $\cH^*$ of the $k$th copy, respectively. By Born's rule on quantum measurement, the (infinitesimal) probabilities describing the possible outcomes of the measurements are given for every $k$ by
\be \label{eq:Born-rule-1-0}
\dd \mathbb{P}_{(\Pi)}^{(\cA)}(X_k=x_k, Y_k=y_k)
=\tr_{\hohc}\lesq{\Pi \ler{\dd E_k(y_k) \otimes \dd E_k^T(x_k)}}.
\ee
By the independence of the measurements on different copies, the joint law $ \mathbb{P}_{(\Pi)}^{(\cA)}$ of the $\R^K$-valued random vectors $X=(X_1, \dots, X_K)$ and $Y=(Y_1, \dots Y_K)$ is given by
$$
\dd \mathbb{P}_{(\Pi)}^{(\cA)}(x_1, \dots, x_K, y_1, \dots, y_K)
=\dd \mathbb{P}_{(\Pi)}^{(\cA)}(X_1=x_1, \dots, X_K=x_K, Y_1=y_1, \dots, Y_K=y_K)
$$
$$
= \prod_{k=1}^K \tr_{\cH \otimes \cH^*}\lesq{\Pi \ler{\dd E_k(y_k) \otimes \dd E_k^T (x_k)}}
$$
\be \label{eq:P-Pi}
=\tr_{\ler{\cH \otimes \cH^*}^{\otimes K}}\lesq{\Pi^{\otimes K} \ler{\dd E_1(y_1) \otimes \dd E_1^T (x_1) \otimes \dots \otimes \dd E_K(y_K) \otimes \dd E_K^T (x_K)}}.
\ee
In analogy with \eqref{eq:cl-ot-prob}, we define the following loss function:
\be \label{eq:quantum-loss-0}
\cS(\cH \otimes \cH^*) \ni \Pi \mapsto \E_{(\Pi)}^{(\cA)}\lesq{c(X,Y)}=\iint_{\R^K \times \R^K} c(x_1, \dots, x_K, y_1, \dots, y_K) \dd \mathbb{P}_{(\Pi)}^{(\cA)}(x_1, \dots, x_K, y_1, \dots, y_K).
\ee
Note that $c$ is a non-negative Borel measurable function on $\R^{2K},$ and $\mathbb{P}_{(\Pi)}^{(\cA)}$ is a genuine (non-negative scalar-valued) probability measure on $\R^{2K},$ and hence the integral on the right-hand-side of \eqref{eq:quantum-loss-0} is well-defined and takes values in $[0,+\infty].$
In view of \eqref{eq:P-Pi}, the loss function $\E_{(\Pi)}^{(\cA)}\lesq{c(X,Y)}$ is equal to
$$
\iint_{\R^K \times \R^K} c(x_1, \dots, x_K, y_1, \dots, y_K)  \tr_{\ler{\cH \otimes \cH^*}^{\otimes K}}\lesq{\Pi^{\otimes K} \ler{\dd E_1(y_1) \otimes \dd E_1^T (x_1) \otimes \dots \otimes \dd E_K(y_K) \otimes \dd E_K^T (x_K)}}.
$$
In transport theory, a usual ``finite energy condition" ensuring that $\iint_{\R^K \times \R^K} c(x,y) \dd \pi(x,y)$ is finite for every coupling of $\mu$ and $\nu$ is that $c(x_1, \dots, x_K,y_1, \dots, y_K) \leq c_X(x_1, \dots ,x_K)+c_Y(y_1, \dots, y_K)$ for some $c_X \in L^1(\R^K, \mu)$ and $c_Y \in L^1(\R^K, \nu)$ --- see, e.g., \cite[Remark 4.5]{Villani2}. Accordingly, here we impose the following finite energy condition. Let $\mu_k$ and $\nu_k$ be the laws of $A_k$ in the states $\rho$ and $\omega$ respectively, that is, $\dd \mu_k(x_k)=\tr_{\cH}[\rho \, \dd E_k(x_k)]$ and $\dd \nu_k(y_k)=\tr_{\cH}[\omega \, \dd E_k(y_k)]$ for every $k \in \{1, \dots,K\},$ and let $c(x,y)$ satisfy $0\leq c(x,y) \leq c_X(x_1, \dots, x_K)+c_Y(y_1, \dots, y_K)$ for some $c_X \in L^1(\R^K, \mu_1 \otimes  \dots \otimes \mu_K)$ and $c_Y \in L^1(\R^K, \nu_1 \otimes \dots \otimes  \nu_K).$ Under this assumption, $\E_{(\Pi)}^{(\cA)}[c(X,Y)]$ is finite for any $\pi \in \cC(\rho, \omega).$ Indeed, in this case
$$
\E_{(\Pi)}^{(\cA)}\lesq{c(X,Y)}
=\iint_{\R^K \times \R^K} c(x_1, \dots, x_K, y_1, \dots, y_K) \prod_{k=1}^K \tr_{\cH \otimes \cH^*}\lesq{\Pi \ler{\dd E_k(y_k) \otimes \dd E_k^T (x_k)}}
$$
$$
\leq \iint_{\R^K \times \R^K} \ler{c_X(x_1, \dots ,x_K)+c_Y(y_1, \dots ,y_K)} \prod_{k=1}^K \tr_{\cH \otimes \cH^*}\lesq{\Pi \ler{\dd E_k(y_k) \otimes \dd E_k^T (x_k)}}
$$
$$
=\int_{\R^K} c_X(x_1, \dots ,x_K) \prod_{k=1}^K \tr_{\cH}[\rho \, \dd E_k(x_k)]
+\int_{\R^K} c_Y(y_1, \dots ,y_K) \prod_{k=1}^K \tr_{\cH}[\omega \, \dd E_k(y_k)]
$$
$$
=\int_{\R^K} c_X \dd (\mu_1 \otimes \dots \otimes \mu_K) + \int_{\R^K} c_Y \dd (\nu_1 \otimes \dots \otimes \nu_K)< \infty.
$$
The above computation encourages us to define the positive and possibly unbounded quantum cost operator $C_{c}^{(\cA)}$ by
\be \label{eq:Q-cost-op-def}
C_{c}^{(\cA)}:=\iint_{\R^K \times \R^K} c(x_1, \dots, x_K, y_1, \dots, y_K) \dd E_1(y_1) \otimes \dd E_1^T (x_1) \otimes \dots \otimes \dd E_K(y_K) \otimes \dd E_K^T (x_K)
\ee
and propose the following quantum optimal transport problem:
\be \label{eq:q-ot-prob-prop}
\text{ minimize } \Pi \mapsto \tr_{\ler{\hohc}^{\otimes K}}\lesq{\Pi^{\otimes K} C_{c}^{(\cA)}}
\ee
where $\Pi$ runs over the set of all couplings of $\rho$ and $\omega$ which are states on $\cH.$ As the cost operator $C_c^{(\cA)}$ acting on its domain contained in $\ler{\hohc}^{\otimes K}$ is unbounded if so is $c$ on $\R^{2K},$ we emphasize that by the formal expression $\tr_{\ler{\hohc}^{\otimes K}}\lesq{\Pi^{\otimes K} C_{c}^{(\cA)}}$ we mean $\lim_{n\to \infty} \tr_{\ler{\hohc}^{\otimes K}}\lesq{\Pi^{\otimes K} \chi_n\ler{C_{c}^{(\cA)}}},$ where $\chi_n: [0, \infty) \rightarrow [0,n]$ is the cut function defined by $\chi_n(x)=\min{\{x,n\}}.$ Note that this limit exists in $[0,+\infty]$ and coincides with $\E_{(\Pi)}^{(\cA)}[c(X,Y)]=\iint_{\R^K \times \R^K}c(x,y) \dd \mathbb{P}_{(\Pi)}^{(\cA)}(x,y).$

\begin{proposition}[Existence of optimal plans] \label{prop:ex-opt-plan}
Let $c: \R^K \times \R^K \rightarrow \R$ be a non-negative and lower semi-continuous function, let $\cA$ be a finite collection of observables on $\cH,$ and $\rho, \omega \in \sh$ given marginals. Then there exists an optimal solution $\Pi_0 \in \cC(\rho, \omega)$ of the optimization problem \eqref{eq:q-ot-prob-prop}.
\end{proposition}

\begin{proof}
Let $\lers{\Pi_n}_{n \in \N} \subset \cC(\rho, \omega)$ be a minimizing sequence for \eqref{eq:q-ot-prob-prop}, that is, 
$$
\lim_{n \to \infty}\tr_{\ler{\hohc}^{\otimes K}}\lesq{\Pi_n^{\otimes K} C_{c}^{(\cA)}}
=\inf\lers{\tr_{\ler{\hohc}^{\otimes K}}\lesq{\Pi^{\otimes K} C_{c}^{(\cA)}} \, \middle| \, \Pi \in \cC(\rho, \omega)}.
$$
Then, according to the compactness result concerning quantum transport plans \cite[Proposition 8]{DPT-AHP}, there is a subsequence $\lers{\Pi_{n_j}}_{j \in \N}$ and an element $\widetilde{\Pi} \in \cC(\rho, \omega)$ such that $\Pi_{n_j} \to \widetilde{\Pi}$ in trace norm. Consequently, $\mathbb{P}_{\ler{\Pi_{n_j}}}^{(\cA)}$ converges to $\mathbb{P}_{\ler{\widetilde{\Pi}}}^{(\cA)}$ set-wise. Indeed, for any Borel set $B \subset \R^{2K}$ the convergence
$$
\mathbb{P}_{\ler{\Pi_{n_j}}}^{(\cA)}(B)=\tr_{(\hohc)^{\otimes K}} \lesq{\ler{\Pi_{n_j}}^{\otimes K} \ler{E_1 \otimes E_1^T \otimes \dots \otimes E_K \otimes E_K^T} (B)} \rightarrow
$$
$$
\rightarrow
\tr_{(\hohc)^{\otimes K}} \lesq{\widetilde{\Pi}^{\otimes K} \ler{E_1 \otimes E_1^T \otimes \dots \otimes E_K \otimes E_K^T} (B)}=\mathbb{P}_{\ler{\widetilde{\Pi}}}^{(\cA)}(B)
$$
follows from the fact that $\tr_{(\hohc)^{\otimes K}}\lesq{\abs{\ler{\Pi_{n_j}}^{\otimes K}-\widetilde{\Pi}^{\otimes K}}}$ tends to $0$ as $j$ tends to infinity. The last step will be the proof of the lower semi-continuity of the cost functional
$$
\gamma \mapsto \iint_{\R^K \times \R^k} c(x,y) \dd \gamma(x,y)
$$
on the set of Borel probabilities on $\R^K \times \R^K$ with respect to the set-wise convergence. This is a standard argument (see, e.g., \cite[Lemma 4.3]{Villani2}) which is based on the lower semi-continuity of $(x,y) \mapsto c(x,y)$ itself. As $c$ is non-negative and lower semi-continuous, it is the pointwise limit of a monotone increasing family of non-negative, continuous and bounded functions $\lers{c_r}_{r \in \N}$ on $\R^K.$ By Beppo Levi's monotone convergence theorem, 
$$
\iint_{\R^K \times \R^K} c(x,y) \dd \bP_{\ler{\widetilde{\Pi}}}^{(\cA)}(x,y)
= \lim_{r \to \infty } \iint_{\R^K \times \R^K} c_r(x,y) \dd \bP_{\ler{\widetilde{\Pi}}}^{(\cA)}(x,y)
$$
$$
=  \lim_{r \to \infty } \ler{ \lim_{j \to \infty}
\iint_{\R^K \times \R^K} c_r(x,y) \dd \bP_{\ler{\Pi_{n_j}}}^{(\cA)}(x,y)} \leq \liminf_{j \to \infty} \iint_{\R^K \times \R^K} c(x,y)  \dd \bP_{\ler{\Pi_{n_j}}}^{(\cA)}(x,y). 
$$
Here we used that the set-wise convergence of probability measures implies the convergence with respect to continuous and bounded test functions.
\end{proof}

A particular emphasis should be given to the case when $c$ is the power of order $p$ of the $l_q$ norm distance on $\R^K,$ that is, 
\be \label{eq:c-l-q-p-norm-def}
c(x_1, \dots, x_K, y_1, \dots, y_K)=
\ler{\sum_{k=1}^K \abs{x_k-y_k}^q}^{\frac{p}{q}},
\ee
where $p>0$ and $q \geq 1.$ In this special case we denote the corresponding cost operator by $C_{p,q}^{(\cA)},$ moreover, as the following calculation shows, the cost operator admits a picturesque closed form that does not involve the spectral resolution in an explicit way:
$$
C_{p,q}^{(\cA)}=\iint_{\R^K \times \R^K} \ler{\sum_{k=1}^K \abs{x_k-y_k}^q}^{\frac{p}{q}} \dd E_1(y_1) \otimes \dd E_1^T (x_1) \otimes \dots \otimes \dd E_K(y_K) \otimes \dd E_K^T (x_K)
$$
$$
=\ler{\sum_{k=1}^K \iint_{\R^K \times \R^K}  \abs{x_k-y_k}^q \dd E_1(y_1) \otimes \dd E_1^T (x_1) \otimes \dots \otimes \dd E_K(y_K) \otimes \dd E_K^T (x_K)}^{\frac{p}{q}}
$$
$$
=\ler{\sum_{k=1}^K \abs{\iint_{\R^K \times \R^K}  \ler{x_k-y_k}\dd E_1(y_1) \otimes \dd E_1^T (x_1) \otimes \dots \otimes \dd E_K(y_K) \otimes \dd E_K^T (x_K)}^q}^{\frac{p}{q}}
$$
\be \label{eq:C-A-p-q-form}
= \ler{\sum_{k=1}^K \ler{\abs{A_k \otimes I_{\cH}^T-I_{\cH} \otimes A_k^T}^q}^{(k)} }^{\frac{p}{q}}
\ee
where $\ler{\abs{A_k \otimes I_{\cH}^T-I_{\cH} \otimes A_k^T}^q}^{(k)}$ is the shorthand for
$$
I_{\hohc} \otimes \dots \otimes I_{\hohc}\otimes  \abs{A_k \otimes I_{\cH}^T-I_{\cH} \otimes A_k^T}^q \otimes I_{\hohc} \otimes \dots \otimes I_{\hohc} 
$$
where the nontrivial component is at position $k.$
\par
Now we are in the position to define quantum $(p,q)$-Wasserstein distances.

\begin{definition} \label{def:p-q-A-dist-def}
Let $\cH$ be a separable Hilbert space, $p>0, q \leq 1,$ let $\cA=\lers{A_1, \dots, A_K}$ be a finite collection of observables on $\cH,$ and let $C_{p,q}^{(\cA)}$ be given as in \eqref{eq:C-A-p-q-form}. We denote the corresponding \emph{quantum $(p,q)$-Wasserstein distance} by $D_{p,q}^{(\cA)}$ and define it the following way:
\be \label{eq:p-q-A-dist-def}
D_{p,q}^{(\cA)}(\rho, \omega)=\ler{\inf\lers{ \tr_{\ler{\cH \otimes \cH^*}^{\otimes K}}\lesq{\Pi^{\otimes K} 
C_{p,q}^{(\cA)}} \, \middle| \, \Pi \in \cS(\hohc), \, \tr_{\cH^*}[\Pi]=\omega, \, \tr_{\cH}[\Pi]=\rho^T}}^{\min\{1,1/p\}}. 
\ee
\end{definition}

\paragraph{{\bf Important special cases.}} In the sequel, we will mainly focus on the $q=p$ case. Under this assumption the cost of a coupling $\Pi \in \cC(\rho, \omega)$ greatly simplifies. Indeed, if $q=p$ then
$$
\tr_{\ler{\cH \otimes \cH^*}^{\otimes K}}\lesq{\Pi^{\otimes K} 
C_{p,p}^{(\cA)}}
=
\tr_{\ler{\cH \otimes \cH^*}^{\otimes K}}\lesq{\Pi^{\otimes K} 
\ler{ \sum_{k=1}^K \ler{\abs{A_k \otimes I_{\cH}^T-I_{\cH} \otimes A_k^T}^p}^{(k)} }}
$$
\be \label{eq:p-cost-op}
=\sum_{k=1}^K \tr_{\ler{\cH \otimes \cH^*}}\lesq{\Pi \abs{A_k \otimes I_{\cH}^T-I_{\cH} \otimes A_k^T}^p}.
\ee
We will study the induced Wasserstein distance $D_{p,p}^{(\cA)}$ and its derived quantities (divergences) thoroughly in the subsequent sections. We close this subsection by noting an easy but potentially relevant fact that confirms the compatibility of the above definition of $(p,q)$-Wasserstein distances with the quadratic distances.
\begin{remark}
In the even more special case $p=q=2$ we get back the definition of the quadratic quantum Wasserstein distance introduced by De Palma and Trevisan in \cite[Definition 8]{DPT-AHP}.
\end{remark}

\subsection{A quantum mechanical optimal transport problem without imposing independence} \label{ss:position-op-approach}

Note that the law $\mathbb{P}_{\Pi} \in \mathrm{Prob}(\R^{2K})$ defined in \eqref{eq:P-Pi} always admits a product structure. Indeed,
$$
\mathbb{P}_{\Pi}=\pi_1 \otimes \dots \otimes \pi_K
$$
where $\pi_k \in \mathrm{Prob}(\R^2)$ is a coupling of the laws of the observable $A_k$ in $\rho \in \cS(\cH)$ and $\omega \in \cS(\cH)$ given by
$$
\dd \pi_k(x,y)=\tr_{\hohc}\lesq{\Pi \ler{\dd E_k(y) \otimes \dd E_k^T(x)}}.
$$
In this subsection, we propose a quantum mechanical analogue of \eqref{eq:cl-ot-prob} relying on the position operator where this product structure is not present. In contrast to the quantum optimal transport problem proposed in the previous subsection where we considered $K$ independent copies of the state of the composite system $\Pi \in \cS(\hohc),$ now we allow correlations. In this regard, the optimal transport problem proposed in this subsection is more general than the previous one. On the other hand, we will restrict our attention to the position operator as observable. The reason is that the spectrum of the position is the whole real line, and hence we will be able to pose a quantum optimal transport problem which reduces to the classical optimal transport problem on $\R^K$ when restricted to states and observables commuting with the position operator. Also, some of the computations presented here are parallel to the ones presented in the previous subsection, still, we keep them for completeness.
\par
In quantum mechanics, the Hilbert space describing a particle moving in $\R^K$ is $L^2(\R^K) \simeq L^2(\R)^{\otimes K}.$
Let $Q_k: L^2(\R^K) \supset \mathrm{dom}(Q_k) \rightarrow L^2(\R^K)$ denote the position operator corresponding to the $k$th coordinate, that is, 
\be \label{eq:q-j-def}
(Q_k \psi)(x_1, \dots, x_K)= x_k \psi(x_1, \dots, x_K) \qquad \ler{\psi \in \mathrm{dom}(Q_k) \subset L^2(\R^K)}
\ee
for all $k \in [K].$ Let us introduce the convention that if $A$ is a linear operator on (its domain contained in) $L^2(\R),$ then $A^{(k)}$ denotes a linear operator on $L^2(\R^K) \simeq L^2(\R)^{\otimes K}$ given by
\be \label{eq:jth-particle-op-def}
A^{(k)}=I_{L^2(\R)} \otimes \dots \otimes  I_{L^2(\R)} \otimes A \otimes I_{L^2(\R)} \otimes \dots \otimes  I_{L^2(\R)}
\ee
where $A$ is the $k$th component of the tensor product on the right-hand-side. Note that with this convention, the $k$th coordinate position operator $Q_k$ defined in \eqref{eq:q-j-def} is exactly $Q^{(k)},$ where $Q$ is the position operator on $L^2(\R)$ defined by $(Q\psi)(x)=x\psi(x).$
\par
Let $E: \cB(\R) \rightarrow \cP(L^2(\R))$ be the spectral measure of the position operator $Q,$ that is, $E(S)=M_{\chi_S},$ where $\chi_S$ is the characteristic function of the Borel set $S \in \cB(\R),$ and $M_f$ is the multiplication with $f$ given by $(M_f \psi)(x)=f(x)\psi(x).$ Then $Q=\int_{\R}\lambda \dd E(\lambda),$ and 
$$
Q_k=Q^{(k)}
=I_{L^2(\R)} \otimes \dots \otimes  I_{L^2(\R)} \otimes Q \otimes I_{L^2(\R)} \otimes \dots \otimes  I_{L^2(\R)}
$$
$$
=I_{L^2(\R)} \otimes \dots \otimes  I_{L^2(\R)} \otimes \int_{\R}\lambda \dd E(\lambda) \otimes I_{L^2(\R)} \otimes \dots \otimes  I_{L^2(\R)}
$$
$$
=\int_{\R}\lambda  I_{L^2(\R)} \otimes \dots \otimes  I_{L^2(\R)} \otimes \dd E(\lambda) \otimes I_{L^2(\R)} \otimes \dots \otimes  I_{L^2(\R)}
=\int_{\R}\lambda \dd E(\lambda)^{(k)}.
$$
Let $\omega$ be a state on $L^2(\R^K),$ and $Y_1, Y_2 \dots, Y_K$ denote the random variables we obtain by measuring $Q_1, Q_2, \dots, Q_K$ simultaneously on the system being in the state $\omega.$ By Born's rule, the (infinitesimal) probabilities describing the possible outcomes of our measurement read as follows:
\be \label{eq:Born-rule-1-b}
\dd \mathbb{P}(Y_1=y_1, Y_2=y_2, \dots, Y_K=y_K)
=\tr_{L^2(\R^K)}\lesq{\omega \ler{\dd E(y_1) \otimes \dd E(y_2) \otimes \dots \otimes \dd E(y_K)}}.
\ee
Similarly, if $X_1, X_2 \dots, X_K$ denote the random variables obtained by measuring $Q_1, Q_2, \dots, Q_K$ in the state $\rho$ simultaneously, then
\be \label{eq:Born-rule-2-0}
\dd \mathbb{P}(X_1=x_1, X_2=x_2, \dots, X_K=x_K)
=\tr_{\ler{L^2(\R^K)}^*}\lesq{\rho^T \ler{\dd E(x_1)^T \otimes \dd E(x_2)^T \otimes \dots \otimes \dd E(x_K)^T}}.
\ee
Here we used that $\tr_{\cH}[\tau H]=\tr_{\cH^*}[(\tau H)^T]=\tr_{\cH^*}[H^T \tau^T]=\tr_{\cH^*}[\tau^T H^T]$ whenever $\tau$ is a trace-class and $H$ is a bounded operator on a Hilbert space $\cH.$
Moreover, if we have a two-particle (both moving in $\R^K$) composite system described by the Hilbert space $L^2(\R^K) \otimes \ler{L^2(\R^K)}^*,$ and the state of the system is $\Pi \in \cS\ler{L^2(\R^K) \otimes \ler{L^2(\R^K)}^*},$ and the random vector $(X_1, X_2, \dots, X_K, Y_1, Y_2, \dots, Y_K) \in \R^{2K}$ denotes the outcome of measuring $Q_1, Q_2, \dots, Q_K$ simultaneously on both particles, then by Born's rule 
$$
\dd \mathbb{P}(Y_1=y_1, Y_2=y_2, \dots, Y_n=y_n, X_1=x_1, X_2=x_2, \dots, X_K=x_K)
$$
\be \label{eq:Born-rule-3-0}
=\tr_{L^2(\R^K) \otimes \ler{L^2(\R^K)}^*}\lesq{\Pi \ler{\dd E(y_1) \otimes \dd E(y_2) \otimes \dots \otimes \dd E(y_K) \otimes \dd E(x_1)^T \otimes \dd E(x_2)^T \otimes \dots \otimes \dd E(x_K)^T}}.
\ee
In analogy with \eqref{eq:cl-ot-prob}, we define the quantum optimal transport problem the following way. Let $\rho, \omega \in \cS(\cH)$ be states, and let $c: \R^K \times \R^K$ be nonnegative and lower semicontinuous. As discussed in general in the introduction (see \eqref{eq:q-coup-def}), in the particular case $\cH=L^2(\R^K)$ the couplings of (or transport plans between) $ \rho$ and $\omega$ are given by
\be \label{eq:q-coup-def-0}
\cC\ler{\rho, \omega}=\lers{\Pi \in \cS\ler{L^2(\R^K) \otimes \ler{L^2(\R^K)}^{*}} \, \middle| \, \tr_{\ler{L^2(\R^K)}^{*}} \lesq{\Pi}=\omega, \,  \tr_{L^2(\R^K)} \lesq{\Pi}=\rho^T}.
\ee
The optimization problem is to
$$
\text{minimize } \Pi \mapsto \mathbb{E}_{(\Pi)} \lesq{c(X_1, \dots,X_K, Y_1, \dots, Y_K)},
$$
where $\mathbb{E}_{(\Pi)}$ denotes the expectation with respect to the randomness driven by the state of the composite system $\Pi$ described by \eqref{eq:Born-rule-3-0}. The quantity to be optimized can be rewritten as
$$
\mathbb{E}_{(\Pi)} \lesq{c(X_1, \dots,X_K, Y_1, \dots, Y_K)}
$$
$$
=\iint_{\R^K \times \R^K} c(x_1,\dots, x_K,y_1, \dots, y_K) \dd \mathbb{P}(Y_1=y_1, Y_2=y_2, \dots, Y_K=y_K, X_1=x_1, X_2=x_2, \dots, X_K=x_K)
$$
$$
=\iint_{\R^K \times \R^K}  c(x_1,\dots, x_K,y_1, \dots, y_K) \tr_{L^2(\R^K) \otimes \ler{L^2(\R^K)}^*}\lesq{\Pi \ler{\dd E(y_1) \otimes  \dots \otimes \dd E(y_K) \otimes \dd E(x_1)^T\otimes \dots \otimes \dd E(x_K)^T}}
$$
$$
= \tr_{L^2(\R^K) \otimes \ler{L^2(\R^K)}^*}\lesq{\Pi \ler{\iint_{\R^K \times \R^K}  c(x_1,\dots, x_K,y_1, \dots, y_K) \dd E(y_1) \otimes  \dots \otimes \dd E(y_K) \otimes \dd E(x_1)^T\otimes \dots \otimes \dd E(x_K)^T}}.
$$
The above computation suggests that the positive semidefinite (and possibly unbounded) cost operator $C_c$ corresponding to the classical cost $c$ should be defined by Borel functional calculus as follows:
\be \label{eq:q-gen-cost-op-def}
C_c:=\iint_{\R^K \times \R^K}  c(x_1,\dots, x_K,y_1, \dots, y_K) \dd E(y_1) \otimes  \dots \otimes \dd E(y_K) \otimes \dd E(x_1)^T\otimes \dots \otimes \dd E(x_K)^T
\ee
With this definition, we propose the following quantum optimal transport problem:
\be \label{eq:q-OT-prob-def-gen}
\text{minimize } \tr_{L^2(\R^K) \otimes \ler{L^2(\R^K)}^*}\lesq{\Pi C_c} \text{ over all } \Pi \in \cC(\rho, \omega).
\ee
The above optimization problem \eqref{eq:q-OT-prob-def-gen} admits an optimal solution, that is, there is a coupling $\Pi^* \in \cC(\rho, \omega)$ that realizes the minimum of \eqref{eq:q-OT-prob-def-gen}. Here we omit the details of the proof of this statement as the argument relies on the compactness of $\cC(\rho, \omega)$ and the lower semi-continuity of the cost functional $\Pi \mapsto \tr_{L^2(\R^K) \otimes \ler{L^2(\R^K)}^*}\lesq{\Pi C_c},$ which means that it essentially coincides with proof of Proposition \ref{prop:ex-opt-plan}.
\par
Just like in the i.i.d. case, where $K$ independent copies of the quantum state $\Pi \in \cS(\hohc)$ were prepared (see Subsection \ref{ss:iid-approach}), if the classical cost function $c$ is chosen to be the $p$th power of the $l_q$ norm (see \eqref{eq:c-l-q-p-norm-def}), then the cost operator $C_c$ given by \eqref{eq:q-gen-cost-op-def} greatly simplifies. In this case,
$$
C_c=\iint_{\R^K \times \R^K} \ler{\sum_{k=1}^K \abs{x_k-y_k}^q}^{\frac{p}{q}} \dd E(y_1) \otimes  \dots \otimes \dd E(y_K) \otimes \dd E(x_1)^T\otimes \dots \otimes \dd E(x_K)^T
$$
$$
=\ler{\sum_{k=1}^K \iint_{\R^K \times \R^K}  \abs{x_k-y_k}^q \dd E(y_1) \otimes  \dots \otimes \dd E(y_K) \otimes \dd E(x_1)^T\otimes \dots \otimes \dd E(x_K)^T}^{\frac{p}{q}}
$$
$$
=\ler{\sum_{k=1}^K \abs{\iint_{\R^K \times \R^K}  \ler{x_k-y_k} \dd E(y_1) \otimes  \dots \otimes \dd E(y_K) \otimes \dd E(x_1)^T\otimes \dots \otimes \dd E(x_K)^T}^q}^{\frac{p}{q}}
$$
$$
=\ler{\sum_{k=1}^K \abs{Q^{(k)} \otimes I_{L^2(\R^K)}^T -I_{L^2(\R^K)} \otimes \ler{Q^{(k)}}^T}^q}^{\frac{p}{q}}.
$$

\section{$p$-Wasserstein distances and divergences}
This section is devoted to the study of $(p,p)$-Wasserstein distances which are defined as optimal values of transport costs of the form \eqref{eq:p-cost-op}. Recall that for $\rho, \omega \in \cS(\cH),$ their couplings are given by \eqref{eq:q-coup-def}.
\par 
Let $\fcoh$ denote the set of all finite collections of observables on the separable complex Hilbert space $\cH.$ That is,
\be \label{eq:fcoh-def}
\fcoh=\lers{\cA \, \middle| \, \cA \subset \cL(\cH)^{sa}, \, \#\{\cA\}<\infty}.
\ee 
In view of \eqref{eq:p-cost-op}, we define the cost operator $C_{\cA,p}$ corresponding to a collection of observables $\cA=\lers{A_1,\dots, A_K} \in \fcoh$ and a finite parameter $p>0$ by
\be \label{eq:cap-def}
C_{\cA,p}:=\sum_{k=1}^K \iint_{\R^2} \abs{x-y}^p \dd E_k(x) \otimes \dd E_k^T (y)
=\sum_{k=1}^K \abs{A_k \otimes I^T -I \otimes A_k^T}^p,
\ee
where $E_k$ is the spectral resolution of $A_k,$ that is, $A_k=\int_{\R} \lambda \dd E_k(\lambda).$ 
Note that the cost operators $C_{\cA,p}$ and $C_{p,p}^{(\cA)}$ are intimately related but not identical. For instance, $C_{p,p}^{\cA}$ acts on $(\hohc)^{\otimes K}$ while $C_{\cA,p}$ acts on $\hohc.$ We recall (see \eqref{eq:p-q-A-dist-def} in Definition \ref{def:p-q-A-dist-def} and also \eqref{eq:p-cost-op}) that for a finite $p>0$, the corresponding $(p,p)$-Wasserstein distance of $\rho$ and $\omega$ is
\be \label{eq:p-W-dist-def}
D_{\cA,p}\ler{\rho, \omega}:=D_{p,p}^{(\cA)}\ler{\rho, \omega}
=\ler{\inf_{\Pi \in \cC(\rho,\omega)}\lers{\tr_{\hohc}\left[\Pi C_{\cA,p}\right]}}^{\min\lers{\frac{1}{p},1}}.
\ee
The quadratic quantum Wasserstein \emph{divergences} suggested in \cite{DPT-lecture-notes} and further studied in \cite{BPTV-metric-24} are given in terms of the quadratic quantum Wasserstein distance the following way:
\be \label{eq:quad-qw-div-def}
d_{\cA,2} (\rho, \omega):=\sqrt{D_{\cA,2}^2(\rho,\omega)-\frac{1}{2}\ler{D_{\cA,2}^2(\rho, \rho)+D_{\cA,2}^2(\omega, \omega)}}.
\ee
Therefore, a natural candidate for being the non-quadratic counterpart of \eqref{eq:quad-qw-div-def} is the following: we define the $(\cA,p)$-Wasserstein \emph{divergence} of $\rho$ and $\omega$ by

\be \label{eq:p-W-div-def}
d_{\cA,p}\ler{\rho, \omega}=\ler{D_{\cA,p}^{\max \lers{p,1}}(\rho, \omega)-\frac{1}{2}\ler{D_{\cA,p}^{\max \lers{p,1}}(\rho, \rho)+D_{\cA,p}^{\max \lers{p,1}}(\omega, \omega)}}^{\min\lers{\frac{1}{p},1}}.
\ee
We note that by the above definition \eqref{eq:p-W-div-def} we have 
\be \label{eq:d-a-p-p-form}
d_{\cA,p}^{\max \lers{p,1}}\ler{\rho, \omega}
=\tr_{\hohc}\left[\ler{\Pi_{*}^{(\rho,\omega)}-\fel\ler{\Pi_{*}^{(\rho,\rho)}+\Pi_{*}^{(\omega,\omega)}}} C_{\cA,p}\right]
\ee
where $\Pi_{*}^{(\alpha,\beta)}$ stands for a $C_{\cA,p}$-optimal coupling of the states $\alpha$ and $\beta.$
Although the definition \eqref{eq:p-W-div-def} is highly natural in view of the definition of the quadratic Wasserstein divergence \cite{DPT-lecture-notes}, 
$$
D_{\cA,p}^p(\rho, \omega)-\frac{1}{2}\ler{D_{\cA,p}^p(\rho, \rho)+D_{\cA,p}^p(\omega, \omega)}<0
$$
may happen for $p>2$ as shown in the following.

\begin{proposition}
For every $p>2$, there exist a finite-dimensional Hilbert space $\cH$, a collection
$\cA\in\fcoh$ with $\#\{\cA\}=1$, and pure states $\rho,\omega\in\cS(\cH)$ such that
$$
d_{\cA,p}^p(\rho,\omega)<0.
$$
In particular, the quantity
$$
D_{\cA,p}^p(\rho,\omega)-\frac12\Big(D_{\cA,p}^p(\rho,\rho)+D_{\cA,p}^p(\omega,\omega)\Big)
$$
need not be nonnegative for $p>2$.
\end{proposition}

\begin{proof}
Let $\cH=\bC^3$ with orthonormal basis $\{e_1,e_2,e_3\}$, and define
$$
Q_i:=|e_i\rangle\langle e_i|, \qquad i=1,2,3.
$$
For some $b_1,b_2>0$ to be fixed later let
$$
A:=0\cdot Q_1+b_1Q_2+(b_1+b_2)Q_3.
$$
We take $\cA=\{A\}$. By \eqref{eq:cap-def}, the corresponding cost operator is
$$
C_{\cA,p}
=\sum_{i,j=1}^3 |a_i-a_j|^p\, Q_i\otimes Q_j^T,
$$
where $a_1=0$, $a_2=b_1$, and $a_3=b_1+b_2$. Expanding this expression, we obtain
\begin{align*}
C_{\cA,p}
&=b_1^p(Q_1\otimes Q_2^T+Q_2\otimes Q_1^T)
 +b_2^p(Q_2\otimes Q_3^T+Q_3\otimes Q_2^T) \\
&\quad +(b_1+b_2)^p(Q_1\otimes Q_3^T+Q_3\otimes Q_1^T).
\end{align*}

Let
$$
\rho=|u\rangle\langle u|, \qquad \omega=|v\rangle\langle v|
$$
be pure states. Pure states admit a unique coupling, the product coupling. Hence, by \eqref{eq:d-a-p-p-form},
$$
d_{\cA,p}^p(\rho,\omega)
=\tr_{\hohc}\Big[\Big(\omega\otimes\rho^T-\frac12(\rho\otimes\rho^T+\omega\otimes\omega^T)\Big)C_{\cA,p}\Big].
$$

Writing
$$
r_i:=\langle u,Q_i u\rangle=|\langle e_i,u\rangle|^2,
\qquad
s_i:=\langle v,Q_i v\rangle=|\langle e_i,v\rangle|^2,
$$
a direct computation yields
\begin{align*}
d_{\cA,p}^p(\rho,\omega)
&=b_1^p\big(r_1s_2+s_1r_2-r_1r_2-s_1s_2\big) \\
&\quad +b_2^p\big(r_2s_3+s_2r_3-r_2r_3-s_2s_3\big) \\
&\quad +(b_1+b_2)^p\big(r_1s_3+s_1r_3-r_1r_3-s_1s_3\big).
\end{align*}

We now choose
$$
u=e_2,
\qquad
v=\frac{e_1+e_3}{\sqrt2}.
$$
Then
$$
(r_1,r_2,r_3)=(0,1,0),
\qquad
(s_1,s_2,s_3)=\Big(\frac12,0,\frac12\Big),
$$
and therefore
$$
d_{\cA,p}^p(\rho,\omega)
=\frac12\,b_1^p+\frac12\,b_2^p-\frac14\,(b_1+b_2)^p.
$$

Finally, choosing $b_1=b_2=1$, we obtain
$$
d_{\cA,p}^p(\rho,\omega)=1-2^{p-2}.
$$
Consequently,
$$
d_{\cA,p}^p(\rho,\omega)\ge 0 \iff p\le 2,
$$
and in particular $d_{\cA,p}^p(\rho,\omega)<0$ for every $p>2$.
\end{proof}

\section{Relations between the sets of all possible cost operators for different values of $p$}
It is one of the key results of \cite{DPT-AHP} (see Theorem 1 there) that in the quadratic case, 
$$
\tr_{\hohc}\left[\ler{\Pi_{*}^{(\rho,\omega)}-\fel\ler{\Pi_{*}^{(\rho,\rho)}+\Pi_{*}^{(\omega,\omega)}}} C_{\cA,2}\right] \geq 0
$$
for every $\rho, \omega \in \cS(\cH)$ and for every $\cA \in \fcoh.$
Therefore, if for some $p>0$ and $\cA \in \fcoh$ the cost operator $C_{\cA,p}$ coincides with the quadratic cost operator $C_{\cA',2}$ for some other $\cA' \in \fcoh,$ then the definition \eqref{eq:p-W-div-def} makes sense. This phenomenon serves as a motivation to study relations between the sets of all possible cost operators for different values of $p.$ In this section we study how the choice of the parameter $p$ affects the set of all possible $p$-cost operators. Throughout this section we assume that the spectra of all observables involved are finite. For a finite $p>0,$ let us introduce
\be \label{eq:c-p-def}
\bC_p(\cH):=\lers{C_{\cA,p} \, \middle| \, \cA \in \fcoh, \, \#\{\mathrm{spec}(A)\}<\infty \text{ for every } A \in \cA}.
\ee 
For $p=\infty$ we define
\be
\bC_\infty(\cH)=\{C\geq 0\;|\;\forall \epsilon>0, \exists p': \forall p\geq p', \exists C'\in \bC_p(\cH): \norm{C-C'}_{op}\leq \epsilon\}
\ee
where $\norm{\cdot}_{op}$ denotes the operator norm on $\cB(\cH).$
\par
Cost operators generated by observables of small spectrum will play an important role in the sequel, so we introduce also the notation

\be \label{eq:c-p-k-def}
\bC_p^{(k)}(\cH):=\lers{C_{\cA,p} \, \middle| \, \cA \in \fcoh \text{ and } \#\lers{\mathrm{spec}(A)}\leq k \text{ for all } A \in \cA} \qquad \ler{k \in \N}.
\ee

The first result says that the parameter value $p=1$ leads to the smallest set of possible cost operators, no matter what the Hilbert space $\cH$ is. The proof heavily relies on the characterization of a notion from graph theory, the cut cone, by $l_1$-embeddability.

\begin{proposition} \label{prop:C1-Cp-rel}
We have 
\be \label{eq:C1-Cp-rel}
\bC_1(\cH) \subseteq \bC_p(\cH)
\ee
for all finite $p>0$ and for every separable Hilbert space $\cH.$
\end{proposition}

\begin{proof}
It suffices to see that for every finite $p>0$ and for every $A \in \cL(\cH)^{sa}$ with finite spectrum there are $B_1, \dots, B_K \in \cL(\cH)^{sa},$ all of them having finite spectrum, such that
\be \label{eq:C1-Cp-rel-key}
\abs{A \otimes I^T-I \otimes A^T}=\sum_{k=1}^K \abs{B_k \otimes I^T - I \otimes B_k^T}^p.
\ee
So let $0<p<\infty$ and $A \in \cL(\cH)^{sa}, \, \#\{\mathrm{spec}(A)\}<\infty$ be fixed, and assume that $A$ admits the spectral decomposition $A=\sum_{r=1}^n \lambda_r Q_r.$ Let $J \subseteq [n]$ where the shorthand $[n]:=\lers{1, \dots, n}$ is used. Let us define 
\be \label{eq:BJ-def}
B_J:= \alpha_{J} \sum_{r \in J} Q_r 
\ee
for every $J\subseteq [n].$ That is, $B_J$ is the spectral projection of $A$ corresponding to the set $\lers{\lambda_r \, \middle| \, r \in J},$ multiplied by a non-negative coefficient $\alpha_J$ to be determined later. With this definition,
$$
\abs{B_J \otimes I^T -I \otimes B_J^T}^p
=\abs{\alpha_J \ler{\sum_{r \in J} Q_r} \otimes \ler{\sum_{s \in [n]} Q_s^T}-\alpha_J\ler{\sum_{r \in [n]} Q_r} \otimes \ler{\sum_{s \in J} Q_s^T}}^p
$$
\be \label{eq:BJ-expand}
= \abs{\alpha_J}^p \sum_{r \in [n]} \sum_{s \in [n]}  \abs{\ind{r \in J}-\ind{s \in J}} Q_r \otimes Q_s^T. 
\ee
On the other hand, a very similar computation yields that
\be \label{eq:A-expand}
\abs{A \otimes I^T -I \otimes A^T}=\sum_{r \in [n]} \sum_{s \in [n]} \abs{\lambda_{r}-\lambda_{s}} Q_r \otimes Q_s^T.
\ee
Therefore,
\be \label{eq:decomp-1}
\abs{A \otimes I^T-I \otimes A^T}=\sum_{J \subseteq [n]} \abs{B_J \otimes I^T - I \otimes B_J^T}^p
\ee
if and only if
\be \label{eq:decomp-2}
\sum_{r \in [n]} \sum_{s \in [n]} \abs{\lambda_{r}-\lambda_{s}} Q_r \otimes Q_s^T
=\sum_{J \subseteq [n]} \sum_{r \in [n]} \sum_{s \in [n]}  \abs{\alpha_J}^p \abs{\ind{r \in J}-\ind{s \in J}} Q_r \otimes Q_s^T
\ee
if and only if
\be \label{eq:decomp-3}
\abs{\lambda_{r}-\lambda_{s}}=\sum_{J \subseteq [n]} \abs{\alpha_J}^p \abs{\ind{r \in J}-\ind{s \in J}}
\ee
holds for every $r,s \in [n].$ The function $[n]\times [n] \ni (r,s) \mapsto \abs{\lambda_r-\lambda_s}$ appearing on the left-hand side of  \eqref{eq:decomp-3} is clearly an $l_1$-embeddable semi-metric, that is, there is a positive integer $m$ and points $u_1, \dots, u_n \in \R^m$ such that $\abs{\lambda_r-\lambda_s}=\norm{u_r-u_s}_{l_1(\R^m)}.$ The trivial choice $m=1$ and $u_r:=\lambda_r$ for $r \in \lers{1, \dots, n}$ does the job. Therefore, by Proposition 4.2.2. of \cite{cut-metric-book}, the function $(r,s) \mapsto \abs{\lambda_{r}-\lambda_{s}}$ belongs to the \emph{cut cone $\mathrm{CUT}_n$}. By the definition of the cut cone, see Eq. (4.1.2.) of \cite{cut-metric-book}, this means that there is an element-wise non-negative vector $\lers{\beta_J}_{J \subseteq [n]}$ such that 
\be
\abs{\lambda_{r}-\lambda_{s}}=\sum_{J \subseteq [n]} \beta_J \abs{\ind{r \in J}-\ind{s \in J}}.
\ee
The choice $\alpha_J:=\beta_J^{1/p}$ completes the proof. 
\end{proof}

The next proposition tells us that observables having at most two different eigenvalues generate the same set of cost operators for any positive parameter value.

\begin{proposition} \label{prop:cost-op-2-level-obs}
For any $p, p'>0,$ and for any separable Hilbert space $\cH,$ we have
\be \label{eq:cost-op-2-level-obs}
\bC_{p}^{(2)}(\cH)=\bC_{p'}^{(2)}(\cH).
\ee
That is, the set of all possible cost operators generated by two-level observables is the same for $p$ and $p'.$
\end{proposition}

In other words, for any $p, p' >0$ and any finite collection $\cA$ of observables on $\cH$ having at most two different eigenvalues there is a collection $\cA'$ with the same property such that
\be \label{eq:cost-dist-equiv}
C_{\cA',p'}=C_{\cA,p} \text{ and hence } D_{\cA',p'}^{\max \lers{p',1}}\ler{\cdot,\cdot} \equiv D_{\cA,p}^{\max \lers{p,1}}\ler{\cdot,\cdot}.
\ee

\begin{proof}
If $A \in \cL(\cH)^{sa}$ and $\#\{\mathrm{spec}(A)\}\leq 2,$ then $A=\lambda Q+ \lambda' (I-Q)$ for some $Q \in \cP(\cH)$ and $\lambda, \lambda' \in \R.$
Therefore, 
$$
\abs{A \otimes I^T - I \otimes A^T}^p
=\left|\lambda Q \otimes Q^T+ \lambda' (I-Q) \otimes Q^T+ \lambda Q \otimes (I-Q)^T+ \lambda' (I-Q) \otimes (I-Q)^T- \right.
$$
$$
\left. - \lambda Q \otimes Q^T- \lambda (I-Q) \otimes Q^T- \lambda' Q \otimes (I-Q)^T-\lambda' (I-Q) \otimes (I-Q)^T\right|^p=
$$
$$
=\abs{\lambda-\lambda'}^p\ler{(I-Q) \otimes Q^T + Q \otimes (I-Q)^T}.
$$
Now let $A'$ be given by $A'=\mu Q.$ Then
$$
\abs{A' \otimes I^T - I \otimes (A')^T}^{p'}=\abs{\mu Q \otimes Q^T+ \mu Q \otimes (I-Q)^T-\mu Q \otimes Q^T-\mu (I-Q) \otimes Q^T}^{p'}=
$$
$$
=\abs{\mu}^{p'} \ler{(I-Q) \otimes Q^T + Q \otimes (I-Q)^T}.
$$
We finish the proof by choosing $\mu:= \abs{\lambda-\lambda'}^{\frac{p}{p'}}.$
\end{proof}

Note that Proposition \ref{prop:cost-op-2-level-obs} has the following immediate consequence for the qubit case, that is, for $\cH=\C^2.$

\begin{corollary} \label{cor:qubit-cost-ops-equal}
For every $0<p,p' <\infty$ we have
\be \label{eq:qubit-cost-ops-equal}
\bC_p(\C^2)=\bC_{p'}(\C^2)
\ee
where $\bC_p(\cH)$ is defined by \eqref{eq:c-p-def}. Consequently,
\be \label{eq:concavity}
D_{\cA,p}^{\max \lers{p,1}}(\rho, \omega)\geq \frac{1}{2}\ler{D_{\cA,p}^{\max \lers{p,1}}(\rho, \rho)+D_{\cA,p}^{\max \lers{p,1}}(\omega, \omega)}
\ee
holds for every $0<p<\infty$ and $\cA \in \fco(\C^2)$ and $\rho, \omega \in \cS(\C^2).$ Therefore, the definition \eqref{eq:p-W-div-def} of the $(\cA,p)$-Wasserstein divergence makes sense for any $\cA \in \fco(\C^2)$ and finite $p>0.$ 
\end{corollary}

We continue with a lemma that will be useful several times in the sequel.

\begin{lemma}\label{lemma-commutativity}
Suppose that for a separable Hilbert space $\cH$ and observables $A$ and $\lers{B_k}_{k=1}^K$ and for finite $p,p'>0$
\be\label{eq-commutativity-lemma}
\abs{A\otimes I^T-I\otimes A^T}^p=\sum_{k=1}^K \abs{B_k\otimes I^T-I\otimes B_k^T}^{p'},
\ee
then $\left[A,B_k\right]=0$ for every $k\in [K].$
\end{lemma}
\begin{proof}
One  has
\be
\abs{A\otimes I^T-I\otimes A^T}^p\geq \abs{B_k\otimes I^T-I\otimes B_k^T}^{p'}\;\forall k\in [K],
\ee
and so
\be
\ker\abs{A\otimes I^T-I\otimes A^T}^p\subseteq \ker\abs{B_k\otimes I^T-I\otimes B_k^T}^{p'}\;\forall k\in [K].
\ee
However, for a vector $v\in \cH$ and any selfadjoint operator $A$ we have that
\be
v\otimes v^T\in \ker\abs{A\otimes I^T-I\otimes A^T}^p\iff v\otimes v^T\in \ker\abs{A\otimes I^T-I\otimes A^T}\iff Av=\lambda v,
\ee
for some $\lambda\in \R$. This is clear from the spectral decomposition
\be
\abs{A\otimes I^T-I\otimes A^T}=\iint_{x<y} \abs{x-y}\left[dE\ler{x}\otimes dE^T\ler{y}+dE\ler{y}\otimes dE^T\ler{x}\right].
\ee
Thus we conclude that if \eqref{eq-commutativity-lemma} holds then any eigenvector of $A$ is also an eigenvector of $B_k$ for every $k\in [K].$
\end{proof}

\begin{remark} \label{remark:commutativity}
There is a physical interpretation to Lemma~\ref{lemma-commutativity}, namely if a Wasserstein cost operator is based on the expected value of some power of the difference of the measurement of a single observable $A$ between the subsystems, then any observable $B_k$ used in constructing the cost operator of $A$ can not distinguish more sensitively between the subsystems than $A$.

If $\dim\cH<\infty$, then a stronger claim holds: any eigensubspace of $A$ must be an eigensubspace of $B_k$, $\forall k\in [r]$.
\end{remark}

We conjecture that the sets of all possible $p$-cost operators are ordered in a nice way: the bigger the parameter, the larger the corresponding set.

\begin{conjecture} \label{conj:C-p-well-ordered}
For any separable Hilbert space $\cH,$ and for any parameter values $0<p\leq p' \leq\infty$ we have the relation
\be \label{eq:C-p-well-ordered}
\bC_{p}(\cH)\subseteq \bC_{p'}(\cH).
\ee
\end{conjecture}

Note that \eqref{eq:C-p-well-ordered} holds for $\cH=\C^2,$ and that an affirmative solution of Conjecture \ref{conj:C-p-well-ordered} would imply that
\be \label{eq:C-p-flat-on-0-1}
\bC_{p}(\cH)=\bC_{p'}(\cH) \text{ for every } \cH \text{ and  every } p, p' \in (0,1]. 
\ee
Indeed, Proposition \ref{prop:C1-Cp-rel} and \eqref{eq:C-p-well-ordered} together would yield that
\be \label{eq:}
\bC_{p}(\cH)\subseteq \bC_{1}(\cH)\subseteq \bC_{p'}(\cH) \text{ and } \bC_{p'}(\cH)\subseteq \bC_{1}(\cH)\subseteq \bC_{p}(\cH).
\ee
An affirmative answer to Conjecture \ref{conj:C-p-well-ordered} would also imply that 
the concavity
$$
D_{\cA,p}^{\max \lers{p,1}}(\rho, \omega)\geq \frac{1}{2}\ler{D_{\cA,p}^{\max \lers{p,1}}(\rho, \rho)+D_{\cA,p}^{\max \lers{p,1}}(\omega, \omega)} \qquad \ler{\rho, \omega \in \sh}
$$
holds for all $0<p\leq 2$ and hence the definition of the quantum $p$-Wasserstein divergence \eqref{eq:p-W-div-def} is meaningful. 
We have the following partial result in the direction of Conjecture \ref{conj:C-p-well-ordered}.

\begin{proposition} \label{prop:c-p-3-smaller}
For any separable Hilbert space $\cH$ and for any $p \in (0,1]$ and $p' \in (0, \infty)$ we have
\be \label{eq:c-p-3-smaller}
\bC_p^{(3)}(\cH) \subseteq \bC_{p'}(\cH).
\ee
\end{proposition}

\begin{proof}
It suffices to show that for every $A \in \cL(\cH)^{sa}$ with at most three eigenvalues we can find $B_1, \dots, B_K \in \cL(\cH)^{sa}$ such that 
\be \label{eq:A-decomp}
\abs{A \otimes I^T-I\otimes A^T}^{p}=\sum_{k=1}^K \abs{B_k \otimes I^T-I\otimes B_k^T}^{p'}.
\ee
Let $A=\sum_{r=1}^3 \lambda_r P_r$ be the spectral decomposition of $A$ and let us try to find the $B_k$'s in the form $B_k=\sum_{r=1}^3 \mu_{r}^{(k)} P_r.$ That is, we have to solve the equation 
\be \label{eq:to-be-solved}
\abs{\lambda_r-\lambda_s}^{p}=\sum_{k=1}^{K} \abs{\mu_r^{(k)}-\mu_s^{(k)}}^{p'} \qquad \ler{r,s \in \lers{1,2,3}}.
\ee
Let us further impose the following relations to simplify the picture: $K:=3$ and
\be \label{eq:simplifying-assumptions}
\mu_2^{(1)}=\mu_1^{(1)}=0, \, \mu_3^{(2)}=\mu_2^{(2)}=0, \, \mu_1^{(3)}=\mu_3^{(3)}=0. 
\ee
So the equation to be solved reduces to
\be \label{eq:to-be-solved-reduced}
\abs{\lambda_1-\lambda_2}^{p}=\abs{\mu_1^{(2)}}^{p'}+\abs{\mu_2^{(3)}}^{p'}, \, 
\abs{\lambda_2-\lambda_3}^{p}=\abs{\mu_2^{(3)}}^{p'}+\abs{\mu_3^{(1)}}^{p'}, \,
\abs{\lambda_3-\lambda_1}^{p}=\abs{\mu_3^{(1)}}^{p'}+\abs{\mu_1^{(2)}}^{p'}
\ee
that admits the solution
$$
\mu_3^{(1)}=\fel \ler{\abs{\lambda_2-\lambda_3}^{p}+\abs{\lambda_3-\lambda_1}^{p}-\abs{\lambda_1-\lambda_2}^{p}}^{1/p'},
$$
$$
\mu_1^{(2)}=\fel \ler{\abs{\lambda_3-\lambda_1}^{p}+\abs{\lambda_1-\lambda_2}^{p}-\abs{\lambda_2-\lambda_3}^{p}}^{1/p'},
$$
\be \label{eq:solution}
\mu_2^{(3)}=\fel \ler{\abs{\lambda_1-\lambda_2}^{p}+\abs{\lambda_2-\lambda_3}^{p}-\abs{\lambda_3-\lambda_1}^{p}}^{1/p'}.
\ee
Note that the right-hand sides in \eqref{eq:solution} make sense as $(\alpha,\beta) \mapsto \abs{\alpha-\beta}^p$ is a metric on $\R$ for $p\leq 1$ and hence satisfies the triangle inequality.
\end{proof}

The following results also point in the direction of an affirmative solution of Conjecture \ref{conj:C-p-well-ordered}.

\begin{proposition}
Let $\cH$ be a separable Hilbert space, then for any finite $p,p'>0$ the following are equivalent:
\begin{enumerate}
    \item \label{item-prop-induction-first} For all observable $A$ on $\cH$ with finite spectrum, there exists a finite collection of observables $\lers{B_k}_{k=1}^r$ on $\cH$ with finite spectrum such that
    \be
    \abs{A\otimes I^T-I\otimes A^T}^p=\sum_{k=1}^r \abs{B_k\otimes I^T-I\otimes B_k^T}^{p'};
    \ee
    \item \label{item-prop-induction-second}For all observable $A\neq 0$ on $\cH$ with finite spectrum, there exists a finite collection of observables $\lers{B_k}_{k=1}^r$ on $\cH$ with finite spectrum, such that
    \be
    \abs{A\otimes I^T-I\otimes A^T}^p-\sum_{k=1}^r \abs{B_k\otimes I^T-I\otimes B_k^T}^{p'}=\abs{A'\otimes I^T-I\otimes A'^T}^p,
    \ee
    for some observable $A'$ on $\cH$, such that $\abs{\mathrm{spec}\ler{A'}}<\abs{\mathrm{spec}\ler{A}}$.
\end{enumerate}
\end{proposition}
\begin{proof}
The direction $\ler{\ref{item-prop-induction-first}}\Rightarrow\ler{\ref{item-prop-induction-second}}$ is obvious, while $\ler{\ref{item-prop-induction-second}}\Rightarrow\ler{\ref{item-prop-induction-first}}$ follows by iteration and by the fact that $\#\lers{\mathrm{spec}\ler{A}}<\infty$.
\end{proof}

\begin{proposition} \label{prop:decomp-smaller-spectrum}
Let $\cH$ be a separable Hilbert space, then for any $p'\geq p>0$ and any $A=\sum_{i=1}^n \lambda_i P_i\neq c \cdot I$ observable on $\cH$ there exists an observable $B=\sum_{i=1}^n \mu_i P_i$ and an observable $A'=\sum_{i=1}^n \lambda'_i P_i$ on $\cH$, such that $\#\lers{\mathrm{spec}\ler{A'}}<\#\lers{\mathrm{spec}\ler{A}}$ and
\be
\ler{\abs{A\otimes I^T-I\otimes A^T}^p-\abs{B\otimes I^T-I\otimes B^T}^{p'}}P_i\otimes P_j^T=\abs{A'\otimes I^T-I\otimes A'^T}^p P_i\otimes P_j^T,
\ee
whenever $i\in\lers{1,2}$ or $j\in\lers{1,2}$ or $i=j$. In particular, for $\cH=\C^3$ we get $\bC_p(\C^3) \subseteq \bC_{p'}(\C^3).$
\end{proposition}
\begin{proof}
Without loss of generality, we can assume that $\lambda_1=0$, since
\be
\abs{\Tilde{A}\otimes I^T-I\otimes \Tilde{A}^T}^p=\abs{A\otimes I^T-I\otimes A^T}^p,
\ee
with $\Tilde{A}=\sum_{i=1}^n \ler{\lambda_i-\lambda_1} P_i$.
Let $\mu_1:=0$ and $\mu_2:=\abs{\lambda_2}^{\frac{p}{p'}}$. With
\be
\abs{A\otimes I^T-I\otimes A^T}^p-\abs{B\otimes I^T-I\otimes B^T}^{p'}=\sum_{i,j=1}^n\ler{\abs{\lambda_i-\lambda_j}^p-\abs{\mu_i-\mu_j}^p}P_i\otimes P_j^T:= \sum_{i,j=1}^n M'_{ij}P_i\otimes P_j^T,
\ee
we already have $M'_{12}=M'_{21}=\abs{\lambda'_1-\lambda'_2}^p$, if we choose $\lambda'_1:=\lambda'_2:=0$. Suppose now that $j>2$, we need to show that $\exists \mu_j,\lambda'_j$ such that
\be
M'_{1j}=M'_{2j}=M'_{j1}=M'_{j2}=\abs{\lambda'_j}^p=\abs{\lambda_j}^p-\abs{\mu_j}^{p'}=\abs{\lambda_2-\lambda_j}^p-\abs{\abs{\lambda_2}^{\frac{p}{p'}}-\mu_j}^{p'}.
\ee
Thus we need to show that
\be
\abs{\lambda_j}^p-\abs{\lambda_2-\lambda_j}^p=\abs{\mu_j}^{p'}-\abs{\abs{\lambda_2}^{\frac{p}{p'}}-\mu_j}^{p'}=:f\ler{\mu_j}
\ee
has a solution in $\mu_j$ with the additional condition, $\abs{\mu_j}^{p'}\leq\abs{\lambda_j}^p$. Then we can let $\lambda'_j:=\ler{\abs{\lambda_j}^p-\abs{\mu_j}^{p'}}^{\frac{1}{p}}$. $f\ler{\mu_j}$ is continuous and so it is sufficient to show that
\be
f\ler{\abs{\lambda_j}^\frac{p}{p'}}\leq\abs{\lambda_j}^p-\abs{\lambda_2-\lambda_j}^p\leq f\ler{-\abs{\lambda_j}^\frac{p}{p'}},
\ee
which is equivalent to
\be
\abs{\abs{\lambda_2}^{\frac{p}{p'}}-\abs{\lambda_j}^{\frac{p}{p'}}}^{p'}\leq\abs{\lambda_2-\lambda_j}^p\leq\abs{\abs{\lambda_2}^{\frac{p}{p'}}+\abs{\lambda_j}^{\frac{p}{p'}}}^{p'},
\ee
which follows from the concavity of $\mathrm{id}^{\frac{p}{p'}}$ on real numbers. Indeed, if $f$ is a concave function, then 
\be
\frac{d}{dx}\ler{f\ler{x_0+x}+f\ler{x_0-x}}=f'\ler{x_0+x}-f'\ler{x_0-x}\leq 0,
\ee
if $x\geq 0$ and so by $x_0:=\frac{\abs{\lambda_2}+\abs{\lambda_j}}{2}$,
\be
\abs{\abs{\lambda_2}^{\frac{p}{p'}}+\abs{\lambda_j}^{\frac{p}{p'}}}^{p'}\geq \abs{0^{\frac{p}{p'}}+\ler{\abs{\lambda_2}+\abs{\lambda_j}}^{\frac{p}{p'}}}^{p'}=\ler{\abs{\lambda_2}+\abs{\lambda_j}}^p\geq \abs{\lambda_2-\lambda_j}^p.
\ee
Whereas, if $f$ is a concave function, then 
\be
\frac{d}{dx}\ler{f\ler{x_1-x}-f\ler{x_0-x}}=f'\ler{x_0-x}-f'\ler{x_1-x}\geq 0,
\ee
if $x_1\geq x_0$ and so by $x_0:=\min\lers{\abs{\lambda_2},\abs{\lambda_j}}$ and $x_1:=\max\lers{\abs{\lambda_2},\abs{\lambda_j}}$,
\be
\abs{\abs{\lambda_2}^{\frac{p}{p'}}-\abs{\lambda_j}^{\frac{p}{p'}}}^{p'}\leq \abs{\abs{\abs{\lambda_2}-\abs{\lambda_j}}^{\frac{p}{p'}}-0^{\frac{p}{p'}}}^{p'}=\abs{\abs{\lambda_2}-\abs{\lambda_j}}^p\leq \abs{\lambda_2-\lambda_j}^p.
\ee
\end{proof}

The following proposition complements Conjecture \ref{conj:C-p-well-ordered} and asserts that the map $p \mapsto \bC_{p}(\cH)$ is not only monotone increasing but strictly monotone increasing on $[1,\infty).$

\begin{proposition}\label{prop:C-p-well-ordered-compl}
If $\mathrm{dim}(\cH)\geq 3$ and $0<p<p'<\infty$ such that also $1<p'$ then $\bC_p(\cH) \nsupseteq \bC_{p'}(\cH).$
\end{proposition}

\begin{proof}
Assume indirectly that $\bC_p(\cH) \supseteq \bC_{p'}(\cH)$ for some $p' \in (1, \infty)$ and $p \in (0,p').$ Let $A=\sum_{r=1}^3 \lambda_r P_r$ where $\lambda_1 <\lambda_2<\lambda_3$ and $\lers{P_r}_{r=1}^3$ is a resolution of the identity. Consider the $p'$-cost operator $\abs{A \otimes I^T-I\otimes A^T}^{p'}$ and note that by the indirect assumption it can be written as
\be \label{eq:ind-ass}
\abs{A \otimes I^T-I\otimes A^T}^{p'}=\sum_{k=1}^K \abs{B_k \otimes I^T-I\otimes B_k^T}^{p}.
\ee
By Lemma \ref{lemma-commutativity} and Remark \ref{remark:commutativity}, the above \eqref{eq:ind-ass} implies that $B_k=\sum_{r=1}^3 \mu_r^{(k)} P_r$ with some real numbers $\mu_1^{(k)},\mu_2^{(k)},\mu_3^{(k)} \in \R$ for every $k \in \lers{1, \dots, K},$ and hence \eqref{eq:ind-ass} yields that
\be \label{eq:lambda-mu}
\abs{\lambda_r-\lambda_s}^{p'}=\sum_{k=1}^{K} \abs{\mu_r^{(k)}-\mu_s^{(k)}}^p \qquad \ler{r,s \in \lers{1,2,3}}.
\ee
In particular, for every real parameter $\tilde{p}$ we have
\be \label{eq:lambda-mu-conseq}
\abs{\lambda_1-\lambda_2}^{p'}+\abs{\lambda_2-\lambda_3}^{p'}-2^{1-\tilde{p}}\abs{\lambda_1-\lambda_3}^{p'}
=\sum_{k=1}^{K} \ler{\abs{\mu_1^{(k)}-\mu_2^{(k)}}^p +\abs{\mu_2^{(k)}-\mu_3^{(k)}}^p - 2^{1-\tilde{p}} \abs{\mu_1^{(k)}-\mu_3^{(k)}}^p}.
\ee
Now let us choose the parameter $\tilde{p} \in \R$ such that $\tilde{p}\geq \max \lers{p,1}$ and $\tilde{p}<p'.$ Such a choice is possible by the assumptions on $p$ and $p'.$ Note furthermore that for $p\geq 1$ the monotonicity and the convexity of the function $t \mapsto t^p$ implies that 
$$
\abs{\mu_1^{(k)}-\mu_3^{(k)}}^p \leq \ler{\abs{\mu_1^{(k)}-\mu_2^{(k)}}+\abs{\mu_2^{(k)}-\mu_3^{(k)}}}^p
=2^p\ler{\fel \ler{\abs{\mu_1^{(k)}-\mu_2^{(k)}}+\abs{\mu_2^{(k)}-\mu_3^{(k)}}}}^p \leq
$$
$$
\leq 2 ^{p-1} \ler{\abs{\mu_1^{(k)}-\mu_2^{(k)}}^p +\abs{\mu_2^{(k)}-\mu_3^{(k)}}^p} \leq 2 ^{\tilde{p}-1} \ler{\abs{\mu_1^{(k)}-\mu_2^{(k)}}^p +\abs{\mu_2^{(k)}-\mu_3^{(k)}}^p}
$$
for every $k \in \lers{1, \dots, K}.$ Therefore, the right-hand side of \eqref{eq:lambda-mu-conseq} is non-negative for every choice of the eigenvalues $\lers{\mu_r^{(k)}}_{r,k}.$ And this is the case in the $p<1$ regime as well: it is a simple consequence of the triangle inequality --- note that $(\alpha,\beta) \mapsto \abs{\alpha-\beta}^p$ is a metric on $\R$ for $p\leq 1$ --- and the fact that $1-\tilde{p}\leq 0$ that 
$$
2^{1-\tilde{p}} \abs{\mu_1^{(k)}-\mu_3^{(k)}}^p \leq \abs{\mu_1^{(k)}-\mu_3^{(k)}}^p \leq \abs{\mu_1^{(k)}-\mu_2^{(k)}}^p +\abs{\mu_2^{(k)}-\mu_3^{(k)}}^p
$$
no matter how we choose the eigenvalues $\lers{\mu_r^{(k)}}_{r,k}.$
\par
However, the left-hand side of \eqref{eq:lambda-mu-conseq} may be negative as the following example shows. Let $\lambda_1=-1, \lambda_2=0, \lambda_3=1.$ Then 
$$
\abs{\lambda_1-\lambda_2}^{p'}+\abs{\lambda_2-\lambda_3}^{p'}-2^{1-\tilde{p}}\abs{\lambda_1-\lambda_3}^{p'}
=1+1-2^{1-\tilde{p}}\cdot 2^{p'}=2-2 \cdot 2^{(p'-\tilde{p})}=2(1-2^{(p'-\tilde{p})})
$$
which is negative as $p'>\tilde{p}.$ A contradiction, as desired.
\end{proof}

Proposition \ref{prop:c-p-3-smaller}, Proposition \ref{prop:decomp-smaller-spectrum}, and Proposition \ref{prop:C-p-well-ordered-compl} have the following consequence on $\cH=\C^3.$ 

\begin{corollary} \label{cor:c-3}
For every $p, p' \in (0, 1]$ one has
$$
\bC_p(\C^3)=\bC_{p'}(\C^3)
$$
and for every $p' \in (1,\infty)$ and $p \in (0,p')$ one has
$$
\bC_p(\C^3) \subseteq \bC_{p'}(\C^3) 
\text{ and }
\bC_p(\C^3) \nsupseteq \bC_{p'}(\C^3).
$$   
\end{corollary}

That is, we have complete understanding of the relations between the sets $\bC_p(\cH)$ for different values of $p$ if $\cH=\C^2$ or $\cH=\C^3.$
\par
We close this section with an argument demonstrating that for any separable Hilbert space $\cH,$ the set of cost operators corresponding to $p'=\infty$ contains the set of cost operators corresponding to $p$ for any finite $p.$

\begin{proposition} \label{prop:C-infty-rules}
For all $p\in (0,\infty)$
\be
\bC_p(\cH) \subseteq \bC_{\infty}(\cH).
\ee
\end{proposition}
\begin{proof}
Let $Q_1$ and $Q_2$ be orthogonal projections over $\cH$ such that $Q_1+Q_2\leq I$. Then if $0\leq\lambda\in\mathbb{R}$, then $\lambda(Q_1\otimes Q_2^T+Q_2\otimes Q_1^T)\in C_\infty(\cH)$. To see this consider $A(p):=-\frac{\sqrt[p]{\lambda}}{2}Q_1+\frac{\sqrt[p]{\lambda}}{2}Q_2$. Then
$$
C_{\{A(p)\},p}=\lambda(Q_1\otimes Q_2^T+Q_2\otimes Q_1^T)+
$$
\be
+\frac{\lambda}{2^p}(Q_1\otimes (I-Q_1-Q_2)^T+(I-Q_1-Q_2)\otimes Q_1^T+Q_2\otimes (I-Q_1-Q_2)^T+(I-Q_1-Q_2)\otimes Q_2^T)
\ee
and
\be
\lim_{p\rightarrow \infty}\norm{C_{\{A(p)\},p}-\lambda(Q_1\otimes Q_2^T+Q_2\otimes Q_1^T)}_{op}=0.
\ee
All $\bC_p(\cH)$ and $\bC_\infty(\cH)$ have a positive linear structure, or in other words, they form a positive cone as is clear from the definition of the costs in \eqref{eq:cap-def}. From this and from the breakdown formula for the costs given in \eqref{eq:BJ-expand} and \eqref{eq:A-expand} it follows that the costs of $\lambda(Q_1\otimes Q_2^T+Q_2\otimes Q_1^T)$ generate all possible costs for all $p>0$.
\end{proof}
Note that Proposition \ref{prop:C1-Cp-rel} and Proposition \ref{prop:C-infty-rules} may be summarized as follows: for any separable Hilbert space $\cH$ and for any $p \in (0, \infty),$ the relation $\bC_1(\cH) \subseteq \bC_p(\cH) \subseteq \bC_{\infty}(\cH)$ holds true.
\section{Triangle inequality is inherited by larger parameter values}
In this section we present another phenomenon that motivates us to study the relation between sets of cost operators, and which can be summarized by saying that the triangle inequality is inherited by larger parameter values as long as the set of all possible cost operators does not increase. The precise statement is formalized in the following Proposition.

\begin{proposition} \label{prop:trineq-hier}
Let $0 < p \leq p' <\infty,$ and let $\cO, \cO' \subseteq \fcoh$ with the property that 
\be \label{eq:cost-opers-relation}
\lers{C_{\cA',p'} \, \middle| \, \cA' \in \cO'} \subseteq \lers{C_{\cA,p} \, \middle| \, \cA \in \cO}.
\ee
Assume that the triangle inequality
\be \label{eq:trineq-p}
d_{\cA,p}(\rho,\omega) \leq d_{\cA,p}(\rho,\tau)+d_{\cA,p}(\tau,\omega) \qquad \ler{\rho, \tau, \omega \in \cS(\cH)}
\ee
holds whenever $\cA \in \cO.$ Then we have the triangle inequality
\be \label{eq:trineq-p-prime}
d_{\cA',p'}(\rho,\omega) \leq d_{\cA',p'}(\rho,\tau)+d_{\cA',p'}(\tau,\omega) \qquad \ler{\rho, \tau, \omega \in \cS(\cH)}
\ee
for all $\cA' \in \cO'.$
\end{proposition}

\begin{proof}
Let $\rho,\omega,\tau \in \sh$ and $\cA' \in \cO'.$ By the definition of quantum Wasserstein divergences \eqref{eq:p-W-div-def}, the left-hand side of \eqref{eq:trineq-p-prime} reads as
\be \label{eq:expand-lhs}\
d_{\cA',p'}(\rho,\omega)=\ler{D_{\cA',p'}^{\max \lers{p',1}}(\rho, \omega)-\frac{1}{2}\ler{D_{\cA',p'}^{\max \lers{p',1}}(\rho, \rho)+D_{\cA',p'}^{\max \lers{p',1}}(\omega, \omega)}}^{\min\lers{1/p',1}}.
\ee
By the assumption \eqref{eq:cost-opers-relation}, there is an $\cA \in \cO$ such that 
\be \label{eq:equal-of-costs}
C_{\cA,p}=C_{\cA',p'}
\ee
and hence by the definition of the quantum $p$-Wasserstein distance \eqref{eq:p-W-dist-def},
\be \label{eq:equal-of-dists}
D_{\cA,p}^{\max \lers{p,1}}\ler{\cdot,\cdot} \equiv D_{\cA',p'}^{\max \lers{p',1}}\ler{\cdot,\cdot}
\ee
on $\sh \times \sh,$ and consequently, by the definition of the quantum $p$-Wasserstein divergence \eqref{eq:p-W-div-def},
\be \label{eq:equal-of-divergences}
d_{\cA,p}^{\max \lers{p,1}}\ler{\cdot,\cdot} \equiv d_{\cA',p'}^{\max \lers{p',1}}\ler{\cdot,\cdot}
\ee 
on $\sh \times \sh.$ Therefore, with this $\cA$ in hand, we can write the right-hand side of \eqref{eq:expand-lhs} as
\be \label{eq:serc}
\ler{
\ler{D_{\cA,p}^{\max \lers{p,1}}(\rho, \omega)-\frac{1}{2}\ler{D_{\cA,p}^{\max \lers{p,1}}(\rho, \rho)+D_{\cA,p}^{\max \lers{p,1}}(\omega, \omega)}}^{\min\lers{1/p,1}}
}^{\frac{\min\lers{1/p',1}}{\min\lers{1/p,1}}}
\ee
which is equal to
$$
\ler{d_{\cA,p}(\rho, \omega)}^{\frac{\min\lers{1/p',1}}{\min\lers{1/p,1}}}
$$
by \eqref{eq:p-W-div-def}.
We have $\min\lers{1/p',1} \leq \min\lers{1/p,1}$ as $p \leq p'$ by assumption, and hence the function
\be \label{eq:power-function}
[0, \infty) \ni t \mapsto t^{\frac{\min\lers{1/p',1}}{\min\lers{1/p,1}}}
\ee
is subadditive. Subadditivity follows from concavity and from the fact that it maps $0$ to itself. By the triangle inequality \eqref{eq:trineq-p}, the obvious monotonicity of \eqref{eq:power-function} and the subadditivity of \eqref{eq:power-function} we get

$$
\ler{d_{\cA,p}(\rho, \omega)}^{\frac{\min\lers{1/p',1}}{\min\lers{1/p,1}}}
\leq \ler{d_{\cA,p}(\rho,\tau)+d_{\cA,p}(\tau,\omega)}^{\frac{\min\lers{1/p',1}}{\min\lers{1/p,1}}}
\leq 
$$
\be \label{eq:subadd}
\leq \ler{d_{\cA,p}(\rho,\tau)}^{\frac{\min\lers{1/p',1}}{\min\lers{1/p,1}}}
+ \ler{d_{\cA,p}(\tau,\omega)}^{\frac{\min\lers{1/p',1}}{\min\lers{1/p,1}}}
\ee
Referring to the definition \eqref{eq:p-W-div-def} of quantum $p$-Wasserstein divergences again, we can write the right-hand side of \eqref{eq:subadd} as 

$$
\ler{
D_{\cA,p}^{\max \lers{p,1}}(\rho, \tau)
-\frac{1}{2}\ler{D_{\cA,p}^{\max \lers{p,1}}(\rho, \rho)+D_{\cA,p}^{\max \lers{p,1}}(\tau, \tau)}
}^{\min\lers{\frac{1}{p'},1}}+
$$
$$
+
\ler{
D_{\cA,p}^{\max \lers{p,1}}(\tau,\omega)
-\frac{1}{2}\ler{D_{\cA,p}^{\max \lers{p,1}}(\tau, \tau)+D_{\cA,p}^{\max \lers{p,1}}(\omega, \omega)}
}^{\min\lers{\frac{1}{p'},1}}
$$
which is equal by \eqref{eq:equal-of-dists} to
$$
\ler{
D_{\cA',p'}^{\max \lers{p',1}}(\rho, \tau)
-\frac{1}{2}\ler{D_{\cA',p'}^{\max \lers{p',1}}(\rho, \rho)+D_{\cA',p'}^{\max \lers{p',1}}(\tau, \tau)}
}^{\min\lers{\frac{1}{p'},1}}+
$$
$$
+
\ler{
D_{\cA',p'}^{\max \lers{p',1}}(\tau, \omega)
-\frac{1}{2}\ler{D_{\cA',p'}^{\max \lers{p',1}}(\tau, \tau)+D_{\cA',p'}^{\max \lers{p',1}}(\omega, \omega)}
}^{\min\lers{\frac{1}{p'},1}}
$$
\be
=d_{\cA',p'}(\rho, \tau)+d_{\cA',p'}(\tau, \omega). 
\ee
So we started from the left-hand side of \eqref{eq:trineq-p-prime} and we arrived at the right-hand side of \eqref{eq:trineq-p-prime} by a chain of estimations from above, hence the proof is complete.
\end{proof}

\section{$p=2$ as a possible threshold for the triangle inequality}

The following observation concerning quadratic Wasserstein divergences formalized in Proposition \ref{prop:2-W-div-on-pure-states} will play an important role in the argument showing that there is no triangle inequality in general for $p$-Wasserstein divergences for parameter values $p<2.$
\par 
The quadratic quantum Wasserstein distance can be written as
$$
D_{\cA,2}^2\ler{\rho, \omega}=\inf_{\Pi \in \cC(\rho,\omega)}
\lers{
\tr_{\hohc}\left[\Pi \ler{\sum_{k=1}^K \ler{A_k\otimes I^T-I \otimes A_k^T}^2}\right]
}=
$$
\be \label{eq:2-Wass-quant-def}
=
\inf_{\Pi \in \cC(\rho,\omega)}
\lers{
\sum_{k=1}^K
\ler{\tr_{\cH}\left[\rho A_k^2+\omega A_k^2\right]-2 \tr_{\hohc} \left[\Pi \cdot A_k \otimes A_k^T \right]}
}.
\ee
There is an explicit one-to-one correspondence between the couplings of $\rho, \omega \in \sh$ and quantum channels sending $\rho$ to $\omega$ \cite{DPT-AHP,DPT-lecture-notes}. An important consequence of this correspondence is that \eqref{eq:2-Wass-quant-def} has an alternative form that refers to channels instead of couplings:
\be \label{eq:qw-dist-channel}
D_{\cA,2}^2\ler{\rho, \omega}
=\inf_{\Phi\in\text{CPTP}(\rho,\omega)}
\lers{
\sum_{k=1}^K
\tr_{\cH}\lesq{\rho A_k^2+\omega A_k^2-2 \rho^{\fel} A_k \rho^{\fel}\Phi^\dagger\ler{A_k}}
}.
\ee
An important feature of the quadratic Wasserstein distances is that the distance of a state $\rho$ from itself (which may be positive) is always realized by the identity channel --- see \cite[Corollary 1]{DPT-AHP}.
Therefore, the quadratic Wasserstein distance of a state $\rho \in \sh$ from itself admits the explicit form  
\be \label{eq:quad-self-dist-explicit}
D_{\cA,2}^2\ler{\rho,\rho}=2 \sum_{k=1}^K\tr_{\cH}\lesq{\rho A_k^2-\ler{\rho^{\fel}A_k}^2}.
\ee
By definition, see \eqref{eq:p-W-div-def}, we have
\be \label{eq:2-W-div-def}
d_{\cA,2}^2(\rho,\omega)=D_{\cA,2}^2(\rho, \omega)-\fel \ler{D_{\cA,2}^2(\rho, \rho)+D_{\cA,2}^2(\omega, \omega)} \qquad \ler{\rho, \omega \in \sh}
\ee 
and taking both \eqref{eq:qw-dist-channel} and \eqref{eq:quad-self-dist-explicit} into account we get that
\be \label{eq:2-W-div-explicit}
d_{\cA,2}^2\ler{\rho, \omega}
=\inf_{\Phi\in\text{CPTP}(\rho,\omega)}
\lers{
\sum_{k=1}^K
\tr_{\cH}\lesq{\ler{\rho^{\fel} A_k}^2+\ler{\omega^{\fel} A_k}^2-2 \rho^{\fel} A_k \rho^{\fel} \Phi^\dagger\ler{A_k}}
}.
\ee
We note that if either $\rho$ or $\omega$ is pure then the only coupling is the tensor product $\omega\otimes\rho^T$ and the only corresponding channel is $\Phi\left(\cdot\right)=\omega\tr\left(\cdot\right)$, with $\Phi^\dagger\left(\cdot\right)=I\tr\omega\left(\cdot\right).$
Consequently, in this case
\be \label{eq:2-W-div-indep-coupling}
d_{\cA,2}^2(\rho,\omega)=\sum_{k=1}^K\ler{\tr_{\cH} \lesq{\ler{\rho^{\fel} A_k}^2}+\tr_{\cH} \lesq{\ler{\omega^{\fel} A_k}^2}-2 \tr_{\cH} \lesq{\rho A_k}\tr_{\cH} \lesq{\omega A_k}}.
\ee
Let us denote by $\cP_1(\cH)$ the set of rank-one ortho-projections on $\cH,$ that is, the set of pure states, and let us turn to the case when both $\rho$ and $\omega$ are pure states. A direct computation shows that $\tr_{\cH}\lesq{\ler{X^{\fel}Y}^2}=\ler{\tr_{\cH}\lesq{XY}}^2$ whenever $X \in \cP_1(\cH)$ and $Y \in \cL(\cH)^{sa}.$ Therefore, \eqref{eq:2-W-div-indep-coupling} leads to the following useful computational rule.

\begin{proposition} \label{prop:2-W-div-on-pure-states}
For any $\cA=\lers{A_1,\dots,A_K} \in \fcoh$ and for all pure states $\rho, \omega \in \cP_1(\cH),$  the quadratic quantum Wasserstein divergence has the following simple form:
\be \label{eq:2-W-div-on-pure-states}
d_{\cA,2}^2(\rho,\omega)=\sum_{k=1}^{K} \ler{\tr_{\cH}\lesq{\rho A_k}-\tr_{\cH}\lesq{\omega A_k}}^2.
\ee
\end{proposition} 
Note that Proposition \ref{prop:2-W-div-on-pure-states} has the following interesting interpretation: the $2$-Wasserstein divergence $d_{\cA,2}(\cdot,\cdot)$ on $\cP_1(\cH) \times \cP_1(\cH)$ is precisely the pull-back of the Euclidean metric on $\R^K$ by the map
\be \label{eq:generalized-Bloch}
\cP_1(\cH) \ni \rho \mapsto \ler{\tr_{\cH}\lesq{\rho A_k}}_{k=1}^K \in \R^K.
\ee
Note furthermore that in the special case $\cH=\C^2$ and $\cA=\lers{\sigma_1,\sigma_2,\sigma_3}$ the vector $\ler{\tr_{\C^2} \lesq{\rho \sigma_j}}_{k=1}^3$ is the Bloch vector of $\rho$ (see \eqref{eq:bloch-vec-def}), so we may consider the vector $\ler{\tr_{\cH}\lesq{\rho A_k}}_{k=1}^K \in \R^K$ as a generalized Bloch vector of $\rho \in \sh$ determined by $\cA=\lers{A_1, \dots,A_K}.$

An interesting consequence of the statements proved in the previous sections and our earlier work \cite{BPTV-metric-24} is that on quantum bits, the triangle inequality holds for all parameters $p \geq 2.$

\begin{theorem} \label{thm:p-trineq-qubit}
Let $\cH=\C^2,$ let $\cA' \in \fco(\C^2)$ be an arbitrary finite collection of observables on $\C^2,$ and let $p' \geq 2.$ Then the triangle inequality
\be \label{eq:p-trineq-qubit}
d_{\cA',p'}(\rho,\omega) \leq d_{\cA',p'}(\rho,\tau)+d_{\cA',p'}(\tau,\omega)
\ee
holds for every $\rho, \tau, \omega \in \cS(\C^2).$
\end{theorem}

\begin{proof}
We know by analytical arguments and numerical tests \cite{BPTV-metric-24} that
\be \label{eq:trineq-quadratic}
d_{\cA,2}(\rho,\omega) \leq d_{\cA,2}(\rho,\tau)+d_{\cA,2}(\tau,\omega) \qquad \ler{\rho, \tau, \omega \in \cS(\cH)}
\ee
for every $\cH$ and $\cA \in \fcoh.$ In particular, \eqref{eq:trineq-quadratic} holds in the special case $\cH=\C^2$ for all $\cA \in \fco(\C^2).$ We know from Corollary \ref{cor:qubit-cost-ops-equal} that $\bC_{p'}(\C^2) \subseteq \bC_2(\C^2)$ for every $p' \geq 2.$ Therefore, if we choose both $\cO$ and $\cO'$ to be $\fco(\C^2)$ and $p=2$ then the assumption \eqref{eq:cost-opers-relation} in Proposition \ref{prop:trineq-hier} is satisfied. Consequently, the desired inequality \eqref{eq:p-trineq-qubit} follows from \eqref{eq:trineq-quadratic} and Proposition \ref{prop:trineq-hier}. 
\end{proof}

On the other hand, we have counterexamples showing that the triangle inequality does not hold in general for $(\cA,p)$-Wasserstein divergences.

\begin{proposition} \label{prop:no-trineq-in-general}
For every $0<p<2$ there is $\cH$ and $\cA \in \fcoh$ and $\rho,\tau,\omega \in \sh$ such that
\be \label{eq:no-trineq-in-general}
d_{\cA,p}(\rho,\omega) > d_{\cA,p}(\rho,\tau)+d_{\cA,p}(\tau, \omega).
\ee
\end{proposition}

\begin{proof}
Let $\cH=\C^2, \, p'=2,$ and $\cA'=\lers{\sigma_1, \sigma_2, \sigma_3},$ where the $\sigma_j$'s are the Pauli operators
\be \label{eq:pauli}
\sigma_1=\lesq{\ba{cc} 0 & 1  \\  1 & 0  \ea}\qquad \sigma_2=\lesq{\ba{cc} 0 & -i  \\  i & 0  \ea}\qquad \sigma_3=\lesq{\ba{cc} 1 & 0  \\  0 & -1  \ea}.
\ee
We recall that the Bloch vector $\bb_{\rho}$ of a state $\rho \in \mathcal{S}(\cH)$ is defined by
\be \label{eq:bloch-vec-def}
\R^3 \ni \bb_{\rho}:=\ler{\tr_{\C^2}\lesq{\rho \sigma_j} }_{j=1}^3
\ee
and the positivity condition $\rho \geq 0$ is equivalent to the Euclidean length of $\bb_{\rho}$ being at most $1.$
Now observe that Proposition \ref{prop:2-W-div-on-pure-states} tells us that 
\be \label{eq:bloch-distance}
d_{\cA', 2}(\rho, \omega)=\abs{\bb_{\rho}-\bb_{\omega}}_2
\ee
whenever $\rho, \omega \in \cP_1(\C^2),$ that is, when both $\rho$ and $\omega$ are pure --- the right-hand side of \eqref{eq:bloch-distance} is the Euclidean distance of the Bloch vectors of $\rho$ and $\omega.$ As it can be seen directly from the geometry of the sphere,
for every $\varepsilon>0$ there exist $\rho, \tau, \omega \in \cP_1(\C^2)$ such that 
\be \label{eq:two-minus-eps}
d_{\cA', 2}(\rho, \tau)=d_{\cA', 2}(\tau, \omega) \text{ and } d_{\cA', 2}(\rho, \omega)> (2-\varepsilon) d_{\cA', 2}(\rho, \tau). 
\ee
By Corollary \ref{cor:qubit-cost-ops-equal}, for every $p \in (0,2)$ there is an $\cA \in \fco(\C^2)$ such that $C_{\cA,p}=C_{\cA',2}$ and hence $d_{\cA,p}^{\max\lers{p,1}}(\cdot,\cdot) \equiv d_{\cA',2}^{2}(\cdot,\cdot)$ on $\cS(\C^2) \times \cS(\C^2).$ Therefore,
\be \label{eq:dapp-1}
d_{\cA,p}^{\max\lers{p,1}}(\rho,\tau)=d_{\cA,p}^{\max\lers{p,1}}(\tau,\omega) \text{ and } d_{\cA,p}^{\max\lers{p,1}}(\rho,\omega)>(2-\varepsilon)^2 d_{\cA,p}^{\max\lers{p,1}}(\rho,\tau).
\ee
Consequently,
\be \label{eq:dapp-2}
d_{\cA,p} (\rho,\omega)> (2-\varepsilon)^{2 \min\lers{\frac{1}{p},1}} \fel \ler{d_{\cA,p}(\rho,\tau)+d_{\cA,p}(\tau,\omega)}.
\ee
So the triangle inequality fails for every $p$ such that
\be \label{eq:suff-for-fail}
(2-\varepsilon)^{2 \min\lers{\frac{1}{p},1}} \cdot \fel \geq 1.
\ee
Taking the logarithm of base $2$ shows that \eqref{eq:suff-for-fail} is equivalent to
\be \label{eq:suff-for-fail-equiv}
 2 \min\lers{\frac{1}{p},1} \cdot \log_{2} (2-\varepsilon) -1 \geq 0,
\ee
that is,
\be \label{eq:suff-for-fail-equiv-2}
\min\lers{\frac{1}{p},1} \geq \frac{1}{2 \log_{2} (2-\varepsilon)}.
\ee
This latter \eqref{eq:suff-for-fail-equiv-2} holds whenever $p\leq 2 \log_{2} (2-\varepsilon).$ The factor $ \log_{2} (2-\varepsilon)$ approaches $1$ from below as $\varepsilon$ goes to $0,$ hence we have examples for the failure of the triangle inequality for every $p<2.$
\end{proof}

\section{Triangle inequality for quantum Wasserstein divergences --- the proof of Theorem \ref{thm:main}}\label{s:main-proof}

This section is dedicated to the study of the triangle inequality for quadratic quantum Wasserstein divergences. The main result here is a proof of the triangle inequality using the sole assumption that an arbitrary one of the three states involved is pure. This is a generalization of the main result of \cite{BPTV-metric-24}, where we proved the triangle inequality in the case when the intermediate state is pure. It is important to note that ten month after the publication of the first version of this paper, Melchior Wirth proved the triangle inequality for all quadratic quantum Wasserstein divergences in full generality, that is, without any restriction concerning the states involved \cite{wirth-triangle}. His technique relies on advanced complex analysis, and hence is rather different from the argument presented below.
\par 
The next statement will be useful in our proof: if $U_k,V_k\in\mathcal{T}_2(\cH)$ are Hilbert-Schmidt operators for $k=1,2,\dots, K,$ then
\be\label{eq:HS-trineq-direct-sum}
\ler{\sum_{k=1}^K\tr_{\cH}\lesq{|U_k|^2}}^{1/2}+\ler{\sum_{k=1}^K\tr_{\cH}\lesq{|V_k|^2}}^{1/2}\ge\ler{\sum_{k=1}^K\tr_{\cH}\lesq{|U_k+V_k|^2}}^{1/2}.
\ee
holds.
Indeed, for the operators $\mathbf{U}=U_1 \oplus U_2 \oplus \dots \oplus U_K$ and $\mathbf{V}=V_1 \oplus V_2 \oplus \dots \oplus V_K \in \mathcal{T}_2\ler{\C^K \otimes \cH}$ the triangle inequality 
$$
\ler{\tr_{\C^K \otimes \cH} \lesq{|\mathbf{U}|^2}}^{1/2}
+\ler{\tr_{\C^K \otimes \cH} \lesq{|\mathbf{V}|^2}}^{1/2}
\geq 
\ler{\tr_{\C^K \otimes \cH}\lesq{|\mathbf{U}+\mathbf{V}|^2}}^{1/2}
$$
with respect to the Hilbert-Schmidt norm holds, which is exactly the same equation as (\ref{eq:HS-trineq-direct-sum}).

\begin{theorem} \label{thm:main}
Let $\cA=\lers{A_k}_{k=1}^K$ be an arbitrary finite collection of observable quantities, and let $d_{\cA,2}$ be the corresponding quadratic quantum Wasserstein divergence given by \eqref{eq:quad-qw-div-def}. Let $\rho, \omega, \tau \in \cS(\cH)$ and assume that $\rho$ or $\tau$ or $\omega$ is a pure state. Then the triangle inequality
\be \label{eq:tri-in}
d_{\cA,2}(\rho, \omega) \leq d_{\cA,2}(\rho,\tau)+d_{\cA,2}(\tau, \omega)
\ee
holds true.
\end{theorem}

\begin{proof}
The case when the intermediate state $\tau$ is pure is covered by \cite{BPTV-metric-24}. Therefore, we consider the case when the state $\rho$ on the left-hand side or the state $\omega$ on the right-hand side is pure. As the inequality \eqref{eq:tri-in} is symmetric, we assume without loss of generality that $\omega \in \cP_1(\cH)$.
The quadratic quantum Wasserstein divergences can be written in the forms
\be \label{eq:d_RO}
d_{\cA,2}^2\ler{\rho, \omega}=
\sum_{k=1}^K
\tr_{\cH}\lesq{\ler{\rho^{1/2} A_k}^2+\ler{\omega^{1/2} A_k}^2-2 \rho^{1/2} A_k \rho^{1/2} \Theta_{opt}^\dagger\ler{A_k}},
\ee

\be \label{eq:d_RT}
d_{\cA,2}^2\ler{\rho, \tau}=
\sum_{k=1}^K
\tr_{\cH}\lesq{\ler{\rho^{1/2} A_k}^2+\ler{\tau^{1/2} A_k}^2-2 \rho^{1/2} A_k \rho^{1/2} \Phi_{opt}^\dagger\ler{A_k}},
\ee
and
\be \label{eq:d_TO}
d_{\cA,2}^2\ler{\tau, \omega}=
\sum_{k=1}^K
\tr_{\cH}\lesq{\ler{\tau^{1/2} A_k}^2+\ler{\omega^{1/2} A_k}^2-2 \tau^{1/2} A_k \tau^{1/2} \Psi_{opt}^\dagger\ler{A_k}},
\ee
where $\Theta_{opt}\in\text{CPTP}(\rho,\omega)$, $\Phi_{opt}\in\text{CPTP}(\rho,\tau)$ and $\Psi_{opt}\in\text{CPTP}(\tau,\omega)$ are optimal channels corresponding the divergences.
Then we have
$$
d_{\cA,2}\ler{\rho, \tau}+d_{\cA,2}\ler{\tau, \omega}=
$$
$$\ler{
\sum_{k=1}^K
\tr_{\cH}\lesq{\ler{\rho^{1/2} A_k}^2+\ler{\tau^{1/2} A_k}^2-2 \rho^{1/2} A_k \rho^{1/2} \Phi_{opt}^\dagger\ler{A_k}}
}^{1/2}+
$$
$$+\ler{
\sum_{k=1}^K
\tr_{\cH}\lesq{\ler{\tau^{1/2} A_k}^2+\ler{\omega^{1/2} A_k}^2-2 \tau^{1/2} A_k \tau^{1/2} \Psi_{opt}^\dagger\ler{A_k}}
}^{1/2}\ge$$
$$\ge \ler{
\sum_{k=1}^K
\tr_{\cH}\lesq{\ler{\rho^{1/2} A_k}^2+\ler{\rho^{1/2} \Phi_{opt}^\dagger\ler{A_k}}^2-2 \rho^{1/2} A_k \rho^{1/2} \Phi_{opt}^\dagger\ler{A_k}}
}^{1/2}+
$$
$$+\ler{
\sum_{k=1}^K
\tr_{\cH}\lesq{\ler{\tau^{1/2} A_k}^2+\ler{\tau^{1/2}\Psi_{opt}^\dagger\ler{A_k}}^2-2 \tau^{1/2} A_k \tau^{1/2} \Psi_{opt}^\dagger\ler{A_k}}
}^{1/2}=
$$
\be \label{eq:square_form}
\ler{
\sum_{k=1}^K
\tr_{\cH}\lesq{\ler{\rho^{1/2} \ler{A_k-\Phi_{opt}^\dagger\ler{A_k}}}^2}}^{1/2}+
\ler{
\sum_{k=1}^K
\tr_{\cH}\lesq{\ler{\tau^{1/2} \ler{A_j-\Psi_{opt}^\dagger\ler{A_j}}}^2}}^{1/2},
\ee
where we have used the inequalities
$$ 
\tr_{\cH}\lesq{\ler{\tau^{1/2} A_k}^2}=tr_{\cH}\ler{\ler{\Phi_{opt}\ler{\rho}}^{1/2} A_k}^2\ge 
\tr_{\cH}\lesq{\ler{\rho^{1/2} \Phi_{opt}^\dagger\ler{A_k}}^2}
$$
and
$$ 
\tr_{\cH}\lesq{\ler{\omega^{1/2} A_k}^2}=tr_{\cH}\ler{\ler{\Psi_{opt}\ler{\tau}}^{1/2} A_k}^2\ge 
\tr_{\cH}\lesq{\ler{\tau^{1/2} \Psi_{opt}^\dagger\ler{A_k}}^2}.
$$
Indeed, for any CPTP map $\Phi$, state $\rho$ and selfadjoint operator $A$, the inequality
\be \label{eq: Lieb}
\tr_{\cH}\lesq{\ler{\Phi\ler{\rho}}^{1/2} A}^2\ge 
\tr_{\cH}\lesq{\ler{\rho^{1/2} \Phi^\dagger\ler{A}}^2}
\ee
is a special case of the monotonicity version of Lieb's concavity --- see, e.g., \cite[Theorem 1]{carlen-lecture-notes}. We can use \eqref{eq: Lieb} for a further estimation in the following way:
$$
\tr_{\cH}\lesq{\ler{\tau^{1/2} \ler{A_k-\Psi_{opt}^\dagger\ler{A_k}}}^2}=
$$
$$
\tr_{\cH}\lesq{\ler{\ler{\Phi_{opt}\ler{\rho}}^{1/2} \ler{A_k-\Psi_{opt}^\dagger\ler{A_k}}}^2}\ge
$$
\be \label{eq:further}
\ge\tr_{\cH}\lesq{\ler{\rho^{1/2} \ler{\Phi_{opt}^\dagger\ler{A_k}-\Phi_{opt}^\dagger\ler{\Psi_{opt}^\dagger\ler{A_k}}}}^2}.
\ee
Substituting \eqref{eq:further} into the  right-hand side (RHS) of \eqref{eq:square_form} we get
$$
RHS\ge 
\ler{
\sum_{k=1}^K
\tr_{\cH}\lesq{\ler{\rho^{1/2} \ler{A_k-\Phi_{opt}^\dagger\ler{A_k}}}^2}}^{1/2}+
$$
$$
+\ler{
\sum_{k=1}^K
\tr_{\cH}\lesq{\ler{\rho^{1/2} \ler{\Phi_{opt}^\dagger\ler{A_k}-\Phi_{opt}^\dagger\ler{\Psi_{opt}^\dagger\ler{A_k}}}}^2}}^{1/2}\ge
$$
\be \label{eq:triangle}
\ge 
\ler{
\sum_{k=1}^K
\tr_{\cH}\lesq{\ler{\rho^{1/2} \ler{A_k-\Phi_{opt}^\dagger\ler{\Psi_{opt}^\dagger\ler{A_k}}}}^2}}^{1/2},
\ee
where we have used the triangle inequality \eqref{eq:HS-trineq-direct-sum} in the last step. Indeed, with the choice
$$
U_k=\rho^{1/4} \ler{A_k-\Phi_{opt}^\dagger\ler{A_k}}\rho^{1/4},
$$
$$
V_k=\rho^{1/4} \ler{\Phi_{opt}^\dagger\ler{A_k}-\Phi_{opt}^\dagger\ler{\Psi_{opt}^\dagger\ler{A_k}}}\rho^{1/4}
$$
for $k=1,2,\dots, K$ the inequality \eqref{eq:HS-trineq-direct-sum} gives the statement.
The optimality of $\Theta_{opt}$ in \eqref{eq:quad-qw-div-def} also means that
\be \label{eq: Theta_opt}
\sup_{\Theta\in\text{CPTP}(\rho,\omega)}
\lers{
\sum_{k=1}^K\tr_{\cH}\lesq{
2 \rho^{1/2} A_k \rho^{1/2} \Theta^\dagger\ler{A_k}}}=
\sum_{k=1}^K\tr_{\cH}\lesq{
2 \rho^{1/2} A_k \rho^{1/2} \Theta_{opt}^\dagger\ler{A_k}}.
\ee
Since $\Psi_{opt}\ler{\Phi_{opt}\ler{\rho}}=\omega$, i.e. $\Psi_{opt}\circ \Phi_{opt}\in\text{CPTP}(\rho,\omega)$ this means that
\be \label{eq:estimation}
\sum_{k=1}^K\tr_{\cH}\lesq{
2 \rho^{1/2} A_k \rho^{1/2} \Theta_{opt}^\dagger\ler{A_k}}\ge 
\sum_{k=1}^K\tr_{\cH}\lesq{
2 \rho^{1/2} A_k \rho^{1/2} \Phi_{opt}^\dagger\ler{\Psi_{opt}^\dagger\ler{A_k}}}.
\ee
Hence we get the following lower bound for the right-hand side of \eqref{eq:triangle}:
$$
\ler{
\sum_{k=1}^K
\tr_{\cH}\lesq{\ler{\rho^{1/2} \ler{A_k-\Phi_{opt}^\dagger\ler{\Psi_{opt}^\dagger\ler{A_k}}}}^2}}^{1/2}=
$$
$$
=\ler{
\sum_{k=1}^K
\tr_{\cH}\lesq{\ler{\rho^{1/2} A_k}^2+\ler{\rho^{1/2}\Phi_{opt}^\dagger\ler{\Psi_{opt}^\dagger\ler{ A_k}}}^2-2 \rho^{1/2} A_k \rho^{1/2} \Phi_{opt}^\dagger\ler{\Psi_{opt}^\dagger\ler{ A_k}}}
}^{1/2}\ge
$$
\be \label{eq:final_est}
\ge
\ler{
\sum_{k=1}^K
\tr_{\cH}\lesq{\ler{\rho^{1/2} A_k}^2+\ler{\rho^{1/2}\Phi_{opt}^\dagger\ler{\Psi_{opt}^\dagger\ler{ A_k}}}^2-2 \rho^{1/2} A_k \rho^{1/2} \Theta_{opt}^\dagger\ler{A_k}}}
^{1/2}
\ee
If $X$ and $Y$ are self-adjoint operators on $\cH,$ and $X\geq 0,$ and $\tr X=1$, then the Cauchy-Schwarz inequality for the Hilbert-Schmidt inner product of the operators $X^{1/2}$ and $X^{1/4}YX^{1/4}$ implies that $\tr_{\cH} \lesq{\ler{X^{1/2}Y}^2} \geq \ler{\tr_{\cH} \lesq{X Y}}^2,$ and hence we get the following lower bound for the second term of \eqref{eq:final_est}:
$$
\tr_{\cH}\lesq{\ler{\rho^{1/2}\Phi_{opt}^\dagger\ler{\Psi_{opt}^\dagger\ler{ A_k}}}^2}
\geq
\ler{\tr_{\cH}\lesq{\rho \, \Phi_{opt}^\dagger\ler{\Psi_{opt}^\dagger\ler{ A_k}}}}^2=
$$
\be
\ler{\tr_{\cH}\lesq{\Psi_{opt}\ler{\Phi_{opt}\ler{ \rho}}A_k}}^2=
\ler{\tr_{\cH}\lesq{\omega A_k}}^2.
\ee
If $\omega$ is pure, then 
$$
\ler{\tr_{\cH}\lesq{\omega A_k}}^2=\tr_{\cH}\lesq{\ler{\omega^{1/2}A_k}^2},
$$
and the right-hand side of \eqref{eq:final_est} is exactly $d_{\cA,2}(\rho, \omega)$.
\end{proof}

\paragraph*{{\bf Acknowledgements.}} We thank Melchior Wirth for his valuable insight concerning the triangle inequality for quadratic quantum Wasserstein divergences and for fruitful discussions on this topic. We are grateful also to Giacomo De Palma and Dario Trevisan for enlightening discussions on possible non-quadratic versions of the optimal transport concept relying on quantum channels. We thank L\'aszl\'o Lov\'asz for suggesting the use of cut cones in the proof of Proposition \ref{prop:C1-Cp-rel}, which is a key idea there.

\begin{small}
\bibliographystyle{plainurl}  
\bibliography{references.bib}
\end{small}

\end{document}